%% file: main.tex
\documentclass[a4paper,UKenglish]{lipics}
\pdfoutput=1

\usepackage{enumerate,amsfonts,graphicx,url,stfloats,subfig}
\usepackage{amssymb,amsmath,color,verbatim,bm,multicol}
\usepackage{stmaryrd}
\usepackage{paralist}
\usepackage{environ,xspace}
\usepackage{thmtools,thm-restate,wrapfig}

\usepackage{tikz}
\usetikzlibrary{arrows,automata,shapes}

\newcommand{\Sskip}{}

\input{macros}

\input{tikzes}

\newtheorem{claim}[theorem]{Claim}

\title{Nested Weighted Limit-Average Automata of Bounded Width}
\author[1]{Krishnendu Chatterjee}
\author[1]{Thomas A. Henzinger}
\author[2]{Jan Otop}

\affil[1]{IST Austria\\
\texttt{\{krish.chat,tah\}@ist.ac.at}}

\affil[2]{University of Wrocław\\
\texttt{jotop@cs.uni.wroc.pl}}

\begin{document}

\maketitle


\input{intro}

\section{Preliminaries}
\label{s:preliminaries}
\input{preliminaries}

\section{Examples}
\label{s:examples}
\input{examples}


\section{Our Results}
\label{s:results}
\input{results}

\subparagraph*{Acknowledgements.}
This research was supported in part by the Austrian Science Fund (FWF) under grants S11402-N23 (RiSE/SHiNE) and Z211-N23 (Wittgenstein Award), 
ERC Start grant (279307: Graph Games), Vienna Science and Technology Fund (WWTF) through project ICT15-003 and
by the National Science Centre (NCN), Poland under grant 2014/15/D/ST6/04543.


\bibliography{papers}

\end{document}

%% file: macros.tex
\usepackage{amsthm}
\usepackage{enumerate,amsfonts,graphicx,url,stfloats,subfig}
\usepackage{amssymb,amsmath,color,verbatim,bm}
\usepackage{stmaryrd,multirow}
\usepackage{paralist}
\usepackage{environ,xspace}
\usepackage{thmtools,thm-restate,wrapfig}
\usepackage{tikz}
\usetikzlibrary{arrows,automata,shapes}

\usepackage[color=light gray]{todonotes}
\presetkeys{todonotes}{inline}{}

\usetikzlibrary{calc}

\newif\ifsavespace
\savespacefalse

\NewEnviron{killcontents}{}
\ifsavespace

\fi

\newcommand{\Paragraph}[1]{\noindent{\textbf{#1}}}
\newcommand{\softParagraph}[1]{\noindent{\emph{#1}}}

\newcommand{\masterA}{{\cal A}_{mas}}
\newcommand{\nestedA}{\mathbb{A}}
\newcommand{\slaveA}{{\mathfrak{B}}}
\newcommand{\nonnestedA}{{\cal A}}

\newcommand{\silent}[1]{\mathsf{sil}({#1})}

\newcommand{\cost}{{C}}

\newcommand{\masterRun}{\Pi}
\newcommand{\slaveRun}{\pi}
\newcommand{\lang}{{\cal L}}

\newcommand{\valueL}[1]{{\cal L}_{{#1}}}

\newcommand{\abs}{\mathop{\mathsf{Abs}}}

\newcommand{\lpair}[2]{\langle {#1}, {#2} \rangle}

\newcommand{\PTIME}{\textsc{PTime}{}}

\newcommand{\PSPACE}{\textsc{PSpace}{}}
\newcommand{\EXPSPACE}{\textsc{ExpSpace}{}}
\newcommand{\NLOGSPACE}{\textsc{NLogSpace}{}}

\newcommand{\Z}{\mathbb{Z}}

\newcommand{\Q}{\mathbb{Q}}

\newcommand{\hide}[1]{}

\definecolor{darkgreen}{RGB}{00,140,00}
\definecolor{darkred}{RGB}{140,0,00}
\definecolor{darkblue}{RGB}{0,0,140}

\newcommand{\tuple}[1]{\langle #1 \rangle}

\newcommand{\buchi}{B\"{u}chi}

\newcommand{\fsum}{\textsc{Sum}}

\newcommand{\fmax}{\textsc{Max}}
\newcommand{\fmin}{\textsc{Min}}

\newcommand{\flimavg}{\textsc{LimAvg}}

\newcommand{\flimsup}{\textsc{LimSup}}

\newcommand{\fsup}{\textsc{Sup}}

\newcommand{\aut}{{\cal A}}

\newcommand{\Acc}{\mathsf{Acc}}

\newcommand{\InfVal}{\mathsf{InfVal}}

\newcommand{\tikzSlaveA}{
\node[circle,draw] (P1) at (0,0) {$q_a$};
\node[circle,draw] (P2) at (1.8,0){$q_F$};
\node[circle,draw, minimum size=0.5cm] (P0) at (1.8,0){};

\draw[->] (P1) to[loop above] node[above,inner sep=5pt] {$(r,1), (\#,1)$} (P1);

\draw[->] (P1) to node[above] {$(g,0)$} (P2);
}

\newcommand{\tikzMasterOne}{
\node[circle, draw] (Q0) at (0,2.5) {$q_0$};

\draw[->] (Q0) to[loop above] node[above] (E1) {$(r,1)$} (Q0);
\draw[->] (Q0) to[loop below] node[below] (E2) {$(\#,2), (g,2)$} (Q0);
}

\newcommand{\tikzOfSlaveRun}{
\foreach \i/\l in {1/r,2/r,3/r,4/{\texttt{\#}},5/r,6/g,7/r,8/g}
{
   \node[rectangle,draw,minimum height=0.5cm] (a\i) at (\i*0.4,0.3) {$\l$};
}
\foreach \i/\j/\l in {1/2/5,2/3/4,3/4/3,5/6/1,7/8/1}
{
   \draw[->] ($(a\i.north) + (0.03,0)$) to[out=90,in=90] node[above] (b\i) {\l} ($(a\j.north)+ (-0.03,0)$) ;
}
\foreach \i/\j/\l in {1/6/1,2/6/2,3/6/3,4/6/4,5/6/5,
                      2/5/1,3/5/2,4/5/3,5/5/4,
                      3/4/1,4/4/2,5/4/3,
                      5/3/1,
                      7/6/1}
{
   \node[rectangle,draw] (c\i\j) at (\i*0.4,\j*0.55+0.2) {$\l$};
}
\foreach \i/\j/\k in {1/6,2/5,3/4,5/3,7/6}
{
   \draw[->] (c\i\j)  to[out=270,in=110] (b\i) ;
}
\foreach \i/\j/\k in {6/6/5,7/5/4,8/4/3,10/3/1,14/6/1}
{
 \node[rectangle,draw,color=blue, thick, minimum width=\k*0.4cm+0.15cm,minimum height=0.55cm,thick] at (\i*0.2,\j*0.55+0.2) {};
}
\foreach \i/\j/\l/\w/\c in {-2.5/6/\slaveA_1/r^3\texttt{\#}rg/3,-2.5/5/\slaveA_1/r^2\texttt{\#}rg/2,-2.5/4/\slaveA_1/r\texttt{\#}rg/1,-2.5/3/\slaveA_1/rg/1,8/6/\slaveA_1/rg/1}
{
   \node[rectangle] (c\i\j) at (\i*0.43+0.25,\j*0.55+0.2) {${\l}(\w)$};
}
}

%% file: tikzes.tex
\newcommand{\pumpingOne}
{
\node[] at (-0.7,0) {$\optWord$};
\node[] at (0.3,-0.6) {$a$};
\node[] at (1.7,-0.6) {$b$};
\node[] at (2.2,-0.6) {$c_1$};
\node[] at (2.7,-0.6) {$c_2$};
\node[] at (3.7,-0.6) {$d$};

\node[draw, rectangle, minimum height=1.8em, minimum width=1.6em] (V) at (0,0) {};
\node[draw, rectangle, minimum height=1.8em, minimum width=4em] (W) at (1,0) {};
\node[draw, rectangle, minimum height=1.8em, minimum width=6em] (W) at (2.75,0) {};

\draw (3.8,0.315) -- (5.5,0.315);
\draw (3.8,-0.315) -- (5.5,-0.315);

\draw[dashed] (0.3,0.35) -- (0.3,1.3);
\draw[dashed] (1.7,0.35) -- (1.7,1.3);
\draw[dashed] (2.2,-0.35) -- (2.2,1.3);
\draw[dashed] (2.7,-0.35) -- (2.7,1.3);

\node[draw, rectangle, minimum height=0.3em, minimum width=6.4em] at (2.55,0.6) {};
\node[draw, rectangle, minimum height=0.3em, minimum width=6.8em] at (2.45,1.2) {};

\node[draw, rectangle, minimum height=0.3em, minimum width=2.2em] at (1.7,0.9) {};
\node[draw, rectangle, minimum height=0.3em, minimum width=0.9em] at (2.3,0.9) {};
\node[draw, rectangle, minimum height=0.3em, minimum width=1.0em] at (2.7,0.9) {};

\node[draw, fill=blue, rectangle, minimum height=0.3em, minimum width=1.3em] at (2.45,0.6) {};
\node[draw, fill=blue, rectangle, minimum height=0.3em, minimum width=1.3em] at (2.45,1.2) {};

\begin{scope}[yshift=-8em]
\node[] at (-0.7,0) {$\optWord$};
\node[] at (0.3,-0.6) {$a$};
\node[] at (1.7,-0.6) {$b$};
\node[] at (2.2,-0.6) {$c_1$};
\node[] at (2.7,-0.6) {$c_2$};

\node[draw, rectangle, minimum height=1.8em, minimum width=1.6em] (V) at (0,0) {};
\node[draw, rectangle, minimum height=1.8em, minimum width=4em] (W) at (1,0) {};
\node[draw, rectangle, minimum height=1.8em, minimum width=12em] (W) at (2.75,0) {};

\draw (3.8,0.315) -- (5.5,0.315);
\draw (3.8,-0.315) -- (5.5,-0.315);

\draw[dashed] (0.3,0.35) -- (0.3,1.3);
\draw[dashed] (1.7,0.35) -- (1.7,1.3);
\draw[dashed] (2.2,-0.35) -- (2.2,1.3);
\draw[dashed] (2.7,-0.35) -- (2.7,1.3);
\draw[dashed] (3.2,-0.35) -- (3.2,1.3);
\draw[dashed] (3.7,-0.35) -- (3.7,1.3);

\node[draw, rectangle, minimum height=0.3em, minimum width=8.4em] at (2.90,0.6) {};
\node[draw, rectangle, minimum height=0.3em, minimum width=8.8em] at (2.8,1.2) {};

\node[draw, rectangle, minimum height=0.3em, minimum width=2.2em] at (1.7,0.9) {};
\node[draw, rectangle, minimum height=0.3em, minimum width=0.9em] at (2.3,0.9) {};
\node[draw, rectangle, minimum height=0.3em, minimum width=1.3em] at (2.75,0.9) {};
\node[draw, rectangle, minimum height=0.3em, minimum width=1.3em] at (3.25,0.9) {};
\node[draw, rectangle, minimum height=0.3em, minimum width=1.0em] at (3.7,0.9) {};

\node[draw, fill=blue, rectangle, minimum height=0.3em, minimum width=1.3em] at (2.45,0.6) {};
\node[draw, fill=blue, rectangle, minimum height=0.3em, minimum width=1.3em] at (2.95,0.6) {};
\node[draw, fill=blue, rectangle, minimum height=0.3em, minimum width=1.3em] at (3.45,0.6) {};

\node[draw, fill=blue, rectangle, minimum height=0.3em, minimum width=1.3em] at (2.45,1.2) {};
\node[draw, fill=blue, rectangle, minimum height=0.3em, minimum width=1.3em] at (2.95,1.2) {};
\node[draw, fill=blue, rectangle, minimum height=0.3em, minimum width=1.3em] at (3.45,1.2) {};

\end{scope}
}

\newcommand{\pumpingTwo}
{
\node[] at (-0.7,0) {$\optWord$};
\node[] at (0.3,-0.6) {$a$};
\node[] at (1.7,-0.6) {$b$};

\node[draw, rectangle, minimum height=1.8em, minimum width=1em] (V) at (0,0) {$\alpha_{\epsilon}$};
\node[draw, rectangle, minimum height=1.8em, minimum width=4em] (W) at (1,0) {$\beta_{\epsilon}$};

\draw (1.5,0.315) -- (4.5,0.315);
\draw (1.5,-0.315) -- (4.5,-0.315);

\draw[dashed] (0.3,0.0) -- (0.3,1.2);
\draw[dashed] (1.7,0.0) -- (1.7,1.2);

\node[draw, rectangle, minimum height=0.3em, minimum width=1.6em] at (0.2,0.6) {};
\node[draw, rectangle, minimum height=0.3em, minimum width=1.5em] at (0.3,0.9) {};

\node[draw, fill=blue, rectangle, minimum height=0.3em, minimum width=0.3em] at (0.42,0.6) {};
\node[draw, fill=blue, rectangle,minimum height=0.3em, minimum width=0.8em] at (0.45,0.9) {};

\node[draw, rectangle, minimum height=0.3em, minimum width=0.7em] at (0.75,0.6) {};
\node[draw, rectangle, minimum height=0.3em, minimum width=1.0em] at (1.1,0.6) {};
\node[draw, rectangle, minimum height=0.3em, minimum width=1.0em] at (1.3,0.9) {};
\node[draw, rectangle, minimum height=0.3em, minimum width=1.0em] at (0.9,0.9) {};

\node[draw, rectangle, minimum height=0.3em, minimum width=1.4em] at (1.65,0.6) {};
\node[draw, rectangle, minimum height=0.3em, minimum width=1.2em] at (1.75,0.9) {};

\node[draw, fill=green, rectangle, minimum height=0.3em, minimum width=0.8em] at (1.85,0.6) {};
\node[draw, fill=green, rectangle,minimum height=0.3em, minimum width=1.0em] at (1.88,0.9) {};

\node[draw, rectangle, minimum height=1.8em, minimum width=1em] (V) at (0,-2) {$\alpha_{\epsilon}$};
\node[draw, rectangle, minimum height=1.8em, minimum width=4em] (W) at (1,-2) {$\beta_{\epsilon}$};
\node[draw, rectangle, minimum height=1.8em, minimum width=4em] (W) at (2.4,-2) {$\beta_{\epsilon}$};
\node[draw, rectangle, minimum height=1.8em, minimum width=4em] (W) at (3.8,-2) {$\beta_{\epsilon}$};

\begin{scope}[yshift=-5.5em]

\node[draw, fill=blue, rectangle, minimum height=0.3em, minimum width=0.3em] at (0.42,0.6) {};
\node[draw, fill=blue, rectangle,minimum height=0.3em, minimum width=0.8em] at (0.45,0.9) {};

\node[draw, rectangle, minimum height=0.3em, minimum width=0.7em] at (0.75,0.6) {};
\node[draw, rectangle, minimum height=0.3em, minimum width=1.0em] at (1.1,0.6) {};
\node[draw, rectangle, minimum height=0.3em, minimum width=1.0em] at (1.3,0.9) {};
\node[draw, rectangle, minimum height=0.3em, minimum width=1.0em] at (0.9,0.9) {};

\node[draw, rectangle, minimum height=0.3em, minimum width=1.4em] at (1.65,0.6) {};
\node[draw, rectangle, minimum height=0.3em, minimum width=1.2em] at (1.75,0.9) {};

\begin{scope}[xshift=4em]
\node[draw, fill=blue, rectangle, minimum height=0.3em, minimum width=0.3em] at (0.42,0.6) {};
\node[draw, fill=blue, rectangle,minimum height=0.3em, minimum width=0.8em] at (0.45,0.9) {};

\node[draw, rectangle, minimum height=0.3em, minimum width=0.7em] at (0.75,0.6) {};
\node[draw, rectangle, minimum height=0.3em, minimum width=1.0em] at (1.1,0.6) {};
\node[draw, rectangle, minimum height=0.3em, minimum width=1.0em] at (1.3,0.9) {};
\node[draw, rectangle, minimum height=0.3em, minimum width=1.0em] at (0.9,0.9) {};

\node[draw, rectangle, minimum height=0.3em, minimum width=1.4em] at (1.65,0.6) {};
\node[draw, rectangle, minimum height=0.3em, minimum width=1.2em] at (1.75,0.9) {};
\end{scope}

\begin{scope}[xshift=8em]
\node[draw, fill=blue, rectangle, minimum height=0.3em, minimum width=0.3em] at (0.42,0.6) {};
\node[draw, fill=blue, rectangle,minimum height=0.3em, minimum width=0.8em] at (0.45,0.9) {};

\node[draw, rectangle, minimum height=0.3em, minimum width=0.7em] at (0.75,0.6) {};
\node[draw, rectangle, minimum height=0.3em, minimum width=1.0em] at (1.1,0.6) {};
\node[draw, rectangle, minimum height=0.3em, minimum width=1.0em] at (1.3,0.9) {};
\node[draw, rectangle, minimum height=0.3em, minimum width=1.0em] at (0.9,0.9) {};

\node[draw, rectangle, minimum height=0.3em, minimum width=1.4em] at (1.65,0.6) {};
\node[draw, rectangle, minimum height=0.3em, minimum width=1.2em] at (1.75,0.9) {};
\end{scope}
\end{scope}
}

%% file: intro.tex
\begin{abstract}
While weighted automata provide a natural framework to express 
quantitative properties, many basic properties like average
response time cannot be expressed with weighted automata.
Nested weighted automata extend weighted automata and consist
of a master automaton and a set of slave automata that are invoked
by the master automaton. 
Nested weighted automata are strictly more expressive than 
weighted automata (e.g., average response time can be expressed
with nested weighted automata), but the basic decision questions 
have higher complexity (e.g., for deterministic automata,
the emptiness question for nested weighted automata is $\PSPACE$-hard, 
whereas the corresponding complexity for weighted automata is $\PTIME$).
We consider a natural subclass of nested weighted automata where
at any point at most a bounded number $k$ of slave automata can be 
active. 
We focus on automata whose master value function is the limit average.
We show that these nested weighted automata with bounded width 
are strictly more expressive than weighted automata
(e.g., average response time with no overlapping requests can be 
expressed with bound $k=1$, but not with non-nested weighted automata).
We show that the complexity of the basic decision problems (i.e., emptiness
and universality) for the subclass with $k$ constant matches the complexity 
for weighted automata.
Moreover, when $k$ is part of the input given in unary we establish 
$\PSPACE$-completeness.
\end{abstract}

\section{Introduction}

\noindent{\em Traditional to quantitative verification.}
In contrast to the traditional view of formal verification that focuses 
on Boolean properties of systems, such as ``every request is eventually 
granted'', quantitative specifications consider properties like ``the long-run 
average success rate of an operation is at least one half'' or 
``the long-run average response time is below a threshold.''
Such properties are crucial for performance related properties,
for resource-constrained systems, such as embedded systems, and 
significant attention has been devoted to 
them~\cite{Droste:2009:HWA:1667106,Chatterjee08quantitativelanguages,DBLP:journals/corr/abs-1007-4018,DrosteR06,DBLP:conf/lics/AlurDDRY13}.

\smallskip\noindent{\em Weighted automata.}
A classical model to express quantitative properties is {\em weighted
automata} that extends finite automata where every transition is
assigned a rational number called a \emph{weight}. 
Each run results in a sequence of weights, and a \emph{value function} aggregates
the sequence into a single value.
For non-deterministic weighted automata, the value of a word
is the infimum value of all runs over the word.  
Weighted automata provide a natural and flexible framework to 
express quantitative\footnote{We use the term ``quantitative'' in a
non-probabilistic sense, which assigns a quantitative value to each
infinite run of a system, representing long-run average or maximal
response time, or power consumption, or the like, rather than taking a
probabilistic average over different runs.}
properties~\cite{Chatterjee08quantitativelanguages}.  
Weighted automata have been studied over finite words with weights 
from a semiring~\cite{Droste:2009:HWA:1667106},
and extended to infinite words with limit averaging or supremum as 
a value function~\cite{Chatterjee08quantitativelanguages,DBLP:journals/corr/abs-1007-4018,Chatterjee:2009:AWA:1789494.1789497}. 
While weighted automata over semirings can express several 
quantitative properties~\cite{DBLP:journals/jalc/Mohri02}, 
they cannot express long-run average properties that weighted automata 
with limit averaging can~\cite{Chatterjee08quantitativelanguages}.
However, even weighted automata with limit averaging cannot express 
the basic quantitative property of average response 
time~\cite[Example~5]{nested}.

\smallskip\noindent{\em Nested weighted automata.}
To express properties like average response time, weighted automata were 
extended to \emph{nested weighted automata (NWA)}~\cite{nested}. 
An NWA consists of a master automaton and a set 
of slave automata. The master automaton runs over infinite input words.
At every transition the master automaton can invoke a slave automaton that runs 
over a finite subword of the infinite word, starting at the position where 
the slave automaton is invoked.
Each slave automaton terminates after a finite number of steps and returns 
a value to the master automaton. 
Each slave automaton is equipped with a value function for finite words, 
and the master automaton aggregates the returned values from slave automata 
using a value function for infinite words.  
For Boolean finite automata, nested automata are as expressive as the 
non-nested counterpart, whereas NWA are strictly more 
expressive than non-nested weighted automata~\cite{nested}.
It has been shown in~\cite{nested} that NWA provide a 
specification framework where many basic quantitative properties, 
which cannot be expressed by weighted automata, can be expressed easily, 
and it provides a natural framework to study quantitative run-time 
verification.

\smallskip\noindent{\em The basic decision questions.}
We consider the basic automata-theoretic decision questions 
of emptiness and universality. 
The importance of these basic questions in the weighted automata
setting are as follows:
(1)~Consider a system modeled by a finite-automata recognizing 
traces of the system and a quantitative property given as a weighted
automaton or NWA. 
Then whether the worst-case (resp., best-case) behavior has the value 
at least $\lambda$ is the emptiness (resp., universality) question 
on the product.
(2)~Problems related to model measuring (that generalizes model checking)
and model repair also reduces to the emptiness problem~\cite{modelmeasuring,nested}.

\smallskip\noindent{\em Complexity gap.} 
In this work we focus on the following classical value functions:
$\flimavg$ for infinite words, which is the long-run average property; 
and $\fsum,\fsum^+$ (where $\fsum^+$ is the sum of absolute values) for 
finite words.
While NWA are strictly more expressive than weighted automata, 
the complexity of the decision questions are either unknown or considerably higher. 
Table~\ref{tab1} (non-bold-faced results) summarizes the existing results 
for weighted automata~\cite{Chatterjee08quantitativelanguages} and NWA~\cite{nested}, for example, for NWA for $\fsum^+$ the 
known bounds are $\EXPSPACE$ and $\PSPACE$-hard, and for $\fsum$ even the decidability
of the basic decision questions is open (or undecidable).
Thus, a fundamental question is whether there exist sub-classes of 
NWA that are strictly more expressive than
weighted automata and yet have better complexity than general
NWA.
We address this question in this paper.

\smallskip\noindent{\em Nested weighted automata with bounded width.}
For NWA, let the maximum number of slave automata that can be active at any point 
be the {\em width} of the automaton. 
In this work we consider a natural special class of NWA,
namely, NWA with bounded width, i.e., 
NWA where at any point at most $k$ slave automata 
can be active. 
For example, the average response time with bounded number of requests
pending at any point can be expressed as NWA with 
bounded width, but not with weighted automata.
Moreover, the class of NWA with bounded width is 
equivalent to automata with monitor counters~\cite{ChatterjeeHO15}, which are automata equipped 
with counters, where at each transition, a counter can be started, terminated, 
or the value of the counter can be increased or decreased. 
The transitions do not depend on the counter values, and hence they are 
referred to as monitor counters. 
The values of the counters when they are terminated gives rise to 
the sequence of weights, which is aggregated into a single value
with the $\flimavg$ value function (see \cite{ChatterjeeHO15}). 
Automata with monitor counters are similar in spirit with the class
of register automata of~\cite{DBLP:conf/lics/AlurDDRY13}.

\begin{table}[t]
\centering
\def\tabcolsep{7pt}
{\small
\begin{tabular}{|c|c|c|c|}
\hline 
 & Deterministic & Nondeterministic & Nondeterministic \\
 & (Emptiness/Universality) & Emptiness & Universality \\
\hline
Weighted aut. & \multicolumn{2}{c|}{$\PTIME$} & Undecidable \\
\hline
NWA & \multicolumn{2}{c|}{$\EXPSPACE$, $\PSPACE$-hard} &  \\
$(\flimavg,\fsum^+)$     & \multicolumn{2}{c|}{{\bf ${\PTIME}$ (width $k$ is constant)}} &  Undecidable\\
			& \multicolumn{2}{c|}{\bf $\PSPACE$-c.  (bounded width)}    & \\ 
\hline 
NWA & \multicolumn{2}{c|}{Open}  & \\
$(\flimavg,\fsum)$     & \multicolumn{2}{c|}{\bf $\PTIME$ (width $k$ is constant)} & Undecidable \\
			& \multicolumn{2}{c|}{\bf $\PSPACE$-c. (bounded width)}      & \\ 
\hline
\end{tabular}
}
\caption{Decidability and complexity of emptiness and universality for weighted and nested weighted
automata with $\flimavg$ value function and $\fsum$ and $\fsum^+$ value 
function for slave automata. 
Our results are bold faced. 
Moreover all $\PTIME$ results become $\NLOGSPACE$-complete when the weights are specified in unary.
}
\label{tab1}
\vspace{-2em}
\end{table}

\noindent{\em Our contributions.} 
Our contributions are as follows (summarized as bold-faced results in Table~\ref{tab1}): 
\begin{compactenum}
\item {\em Constant width.} 
We show that the emptiness problem (resp., the emptiness and the universality problems) 
for non-deterministic (resp., deterministic) NWA with 
constant width (i.e., $k$ is constant) can be solved in polynomial time
and is $\NLOGSPACE$-complete when the weights are specified in unary. 
Thus we achieve the same complexity as weighted automata for a much more
expressive class of quantitative properties.

\item {\em Bounded width.}
We show that the emptiness problem (resp., the emptiness and the universality problems) 
for non-deterministic (resp., deterministic) NWA with 
bounded width (i.e., $k$ is part of input given in unary) is $\PSPACE$-complete.
Thus we establish precise complexity when $k$ is part of input given in unary.

\item {\em Deciding width.} 
We show that checking whether a given NWA has width $k$ can be solved in polynomial time for constant $k$
and in $\PSPACE$ if $k$ is given in the input (Theorem~\ref{th:checkingWidth}).
\end{compactenum}

\noindent{\em Technical contributions.}
Our main technical contributions for deterministic $(\flimavg;\fsum)$-automata are as follows. 
\begin{compactenum}
\item {\em Infinite infimum}. 
We first identify a necessary and sufficient condition for the infimum value over all words 
to be~$-\infty$, and show that this condition can be checked efficiently.

\item {\em Lasso-approximation}. 
We show that if the above condition does not hold, then the infimum over all words can be 
approximated by lasso words, i.e., words of the form $v u^{\omega}$. 
Moreover, we show that the infimum value is achieved with words where the slave automata
runs for short length relative to the point of the invocation, and hence the partial averages converge.

\item {\em Reduction to width $1$}. Using the lasso-approximation we reduce the emptiness 
problem of width bounded by $k$ to the corresponding problem 
of width $1$. 
We show that the case of width~1 can be solved using standard techniques. 
\end{compactenum}

\smallskip\noindent{\em Related works.}
Weighted automata over finite words have been extensively studied, 
the book~\cite{Droste:2009:HWA:1667106} provides an excellent collection 
of results.
Weighted automata on infinite words have been studied 
in~\cite{Chatterjee08quantitativelanguages,DBLP:journals/corr/abs-1007-4018,DrosteR06}.
The extension to weighted automata with monitor counters over finite words has been 
considered as cost register automata in~\cite{DBLP:conf/lics/AlurDDRY13}.
A version of nested weighted automata over finite words has been 
studied in~\cite{bollig2010pebble}, and nested weighted automata over 
infinite words has been studied in~\cite{nested}.
Several quantitative logics have also been studied, 
such as~\cite{BokerCHK14,BouyerMM14,AlmagorBK14}. 
In this work we consider a subclass of nested weighted automata which is 
strictly more expressive than weighted automata yet achieve the same complexity 
for the basic decision questions.
Probabilistic models (such as Markov decision processes) with quantitative properties 
(such as limit-average or discounted-sum) have also been extensively studied for 
single objectives~\cite{filar,Puterman}, and for multiple objectives and their 
combinations~\cite{CMH06,Cha07,CFW13,BBCFK11,CKK15,Forejt,FKN11,CD11,Baier-CSL-LICS-1,Baier-CSL-LICS-2}.
While NWA with bounded width have been studied under probabilistic semantics~\cite{ChatterjeeHO15},
the basic automata theoretic decision problems have not been studied for them.


%% file: preliminaries.tex
\newcommand{\expected}{\mathbb{E}}
\newcommand{\distrib}{\mathbb{D}}
\newcommand{\prob}{\mathbb{P}}
\newcommand{\threshold}{\boldsymbol{\theta}}
\newcommand{\markov}{\mathcal{M}}
\newcommand{\run}{\pi}
\newcommand{\weightedRun}{\pi^W}
\newcommand{\calU}{\mathcal{U}}

\newcommand{\blank}{\texttt{\#}}
\newcommand{\delim}{\ensuremath{\$}}

\subsection{Words and automata}

\Paragraph{Words}.
We consider a finite \emph{alphabet} of letters $\Sigma$.
A \emph{word} over $\Sigma$ is a (finite or infinite) sequence of letters from $\Sigma$.
We denote the $i$-th letter of a word $w$ by $w[i]$, and for $i < j$ we
have that $w[i,j]$ is the word $w[i] w[i+1] \ldots w[j]$.
The length of a finite word $w$ is denoted by $|w|$; and the length of an infinite word 
$w$ is $|w| = \infty$.
For an infinite word $w$, thus $w[i,\infty]$ is the suffix of the word 
with first $i-1$ letters removed.
 

\Sskip
\Paragraph{Labeled automata}. For a set $X$, an \emph{$X$-labeled automaton} $\aut$ is a tuple
$\tuple{\Sigma, Q, Q_0, \delta, F, \cost}$, where
(1)~$\Sigma$ is the alphabet, 
(2)~$Q$ is a finite set of states, 
(3)~$Q_0 \subseteq Q$ is the set of initial states, 
(4)~$\delta \subseteq Q \times \Sigma \times Q$ is a transition relation,
(5)~$F$ is a set of accepting states,
and 
(6)~$\cost : \delta \mapsto X$ is a labeling function.
A labeled automaton $\tuple{\Sigma, Q, q_0, \delta, F, \cost}$ is 
\emph{deterministic} if and only if 
$\delta$ is a function from $Q \times \Sigma$ into $Q$ 
and $Q_0$ is a singleton. 
In definitions of deterministic labeled automata we omit curly brackets in the description of $Q_0$
and write $\tuple{\Sigma, Q, q_0, \delta, F, \cost}$.

\Sskip
\Paragraph{Semantics of (labeled) automata}. 
A \emph{run} $\run$ of a (labeled) automaton $\aut$ on a word $w$ is a sequence of states
of $\aut$ of length $|w|+1$  
such that $\run[0]$ belong to the initial states of $\aut$
and for every $0 \leq i \leq |w|-1$ we have $(\pi[i], w[i], \pi[i+1])$  is a transition of $\aut$.
A run $\pi$ on a finite word $w$ is \emph{accepting} iff the last state $\pi[|w|]$ of the run 
is an accepting state of $\aut$.
A run $\pi$ on an infinite word $w$ is \emph{accepting} iff some accepting state of $\aut$ occurs
infinitely often in $\pi$. 
For an automaton $\aut$ and a word $w$, we define $\Acc(w)$ as the set of accepting runs on $w$.
Note that for deterministic automata, every word $w$ has at most one accepting run ($|\Acc(w)| \leq 1$).

\Sskip
\Paragraph{Weighted automata}.
A \emph{weighted automaton} is a $\Z$-labeled automaton, where $\Z$ is the set of integers. 
The labels are called \emph{weights}. 

\Sskip
\Paragraph{Semantics of weighted automata}. 
We define the semantics of weighted automata in two steps. First, we define the value of a 
run. Second, we define the value of a word based on the values of its runs.
To define values of runs, we will consider  \emph{value functions} $f$ that 
assign real numbers to sequences of rationals.
Given a non-empty word $w$, every run $\pi$ of $\aut$ on $w$ defines a sequence of weights 
of successive transitions of $\aut$, i.e., 
$\cost(\pi)=(\cost(\pi[i-1], w[i], \pi[i]))_{1\leq i \leq |w|}$; 
and the value $f(\pi)$ of the run $\pi$ is defined as $f(\cost(\pi))$.
We denote by $(\cost(\pi))[i]$ the weight of the $i$-th transition,
i.e., $\cost(\pi[i-1], w[i], \pi[i])$.
The value of a non-empty word $w$ assigned by the automaton $\aut$, denoted by  $\valueL{\aut}(w)$,
is the infimum of the set of values of all {\em accepting} runs;
i.e., $\inf_{\pi \in \Acc(w)} f(\pi)$, and we have the usual semantics that infimum of an
empty set is infinite, i.e., the value of a word that has no accepting run is infinite.
Every run $\pi$ on an empty word has length $1$ and the sequence $\cost(\pi)$ is empty, hence 
we define the value $f(\pi)$ as an external (not a real number) value $\bot$. 
Thus, the value of the empty word is either $\bot$, if the empty word is accepted by $\aut$, or $\infty$ 
otherwise.
To indicate a particular value function $f$ that defines the semantics,
we will call a weighted automaton $\aut$ an $f$-automaton. 

\Sskip
\Paragraph{Value functions}.
For finite runs we consider the following classical value functions: for runs of length $n+1$ we have
\begin{compactitem}
\item \emph{Sum, absolute sum:} the sum function 
$\fsum(\pi) = \sum_{i=1}^{n} (\cost(\pi))[i]$, 
the absolute sum 
$\fsum^+(\pi) = \sum_{i=1}^{n} \abs((\cost(\pi))[i])$, where $\abs(x)$  is the absolute value of $x$,
\end{compactitem}
For infinite runs we consider:
\begin{compactitem}
\item {\em Limit average:} $\flimavg(\pi) = \liminf\limits_{k \rightarrow \infty} \frac{1}{k} \cdot \sum_{i=1}^{k} (\cost(\pi))[i]$.
\end{compactitem}

\Sskip
\Paragraph{Silent moves}. Consider a $(\Z \cup \{ \bot\})$-labeled automaton. We can consider such an automaton as an extension
of a weighted automaton in which transitions labeled by $\bot$ are \emph{silent}, i.e., they do not contribute to 
the value of a run. Formally, for every function $f \in \InfVal$ we define
$\silent{f}$ as the value function that applies $f$ on sequences after removing $\bot$ symbols.
The significance of silent moves is as follows: it allows to ignore transitions, and thus provide
robustness where properties could be specified based on desired events rather than steps.

\subsection{Nested weighted automata}
In this section we describe nested weighted automata introduced in~\cite{nested},
and closely follow the description of~\cite{nested}.
For more details and illustration of such automata we refer the reader 
to~\cite{nested}.
We start with an informal description.

\smallskip\noindent{\em Informal description.}
A \emph{nested weighted automaton} (NWA) consists of a labeled automaton over infinite words, 
called the \emph{master automaton}, a value function $f$ for infinite words,  
and a set of weighted automata over finite words, called \emph{slave automata}. 
A nested weighted automaton can be viewed as follows: 
given a word, we consider the run of the master automaton on the word,
but the weight of each transition is determined by dynamically running 
slave automata; and then the value of a run is obtained using the 
value function $f$.
That is, the master automaton proceeds on an input word as an usual automaton, 
except that before it takes a transition, it can start a slave automaton 
corresponding to the label of the current transition. 
The slave automaton starts at the current position of the word of the master automaton 
and works on some finite part of the input word. Once a slave automaton finishes,
it returns its value to the master automaton, which treats the returned
value as the weight of the current transition that is being executed.
Note that for some transitions the master automaton might not invoke any 
slave automaton, and which corresponds to \emph{silent} transitions.
If one of slave automata rejects, the nested weighted automaton rejects.
We first present an example and then the formal definition.

\begin{example}[Average response time]
\label{ex:ART}
Consider an alphabet $\Sigma$ consisting of requests $r$, grants $g$, 
and null instructions $\blank$.
The average response time (ART) property asks for the average number of 
instructions between any request and the following grant. 
An NWA computing the average response time is depicted in Fig.~\ref{fig:ART}.
At every position with letter $r$ the master automaton $\masterA$ of 
$\nestedA$ invokes the slave automaton $\slaveA_1$, which computes the number 
of letters from its initial position to the first following grant. 
The automaton $\slaveA_1$ is a $\fsum^+$-automaton.
On letters $\blank$ and $g$ the automaton $\masterA$ invokes the slave 
automaton $\slaveA_2$, which is a dummy automaton, i.e., it immediately accepts 
and returns no weight. 
Invoking such a dummy automaton corresponds to taking a silent transition.
Thus, the sequence of values returned by slave automata ($5 4 3 1 1 \ldots$ 
in Fig.~\ref{fig:ART}), is the sequence of response times for each request. 
Therefore, the averages of these values is precisely the average response time;
in the NWA $\nestedA$, the value function $f$ is $\flimavg$.
Also this property cannot be expressed by a non-nested automaton:
a quantitative property is a function from words to reals, and 
as a function the range of non-nested $\flimavg$-automata is bounded,
whereas the ART can have unbounded values (for details see~\cite{nested}).
\begin{figure}[h]
\begin{tikzpicture}
\begin{scope}[xshift=2cm,yshift=3cm]
\tikzSlaveA
\node at (0.8,-0.5) {$\slaveA_1$};
\end{scope}
\begin{scope}[xshift=0cm,yshift=0cm]
\tikzMasterOne
\node at (0,0.7) {$\masterA$};
\end{scope}

\node[circle,draw, minimum size=0.75cm] (S1) at (2.8,1.5){};
\node[circle,draw, minimum size=0.5cm] (S0) at (2.8,1.5){};

\draw[dashed,->] (E1) to (P1);
\draw[dashed,->] (E2) to (S1);

\node at (2.8,0.8) {$\slaveA_2$};

\begin{scope}[xshift=8cm,yshift=0cm]
\tikzOfSlaveRun
\end{scope}
\end{tikzpicture}

\caption{An NWA $\nestedA$ computing ART. The master automaton $\masterA$ 
and slave automata $\slaveA_1, \slaveA_2$ are on the left.
A part of a run of $\nestedA$ on word $rrr\blank rgrg\ldots$ is presented on 
the right. }
\label{fig:ART}
\end{figure}
\end{example}


\Sskip
\Paragraph{Nested weighted automata}. 
A \emph{nested weighted automaton} (NWA) is a tuple $\tuple{\masterA; f; \slaveA_1, \ldots, \slaveA_l}$, where
(1)~$\masterA$, called the \emph{master automaton}, is a $\{1, \ldots, k\}$-labeled automaton over infinite words 
(the labels are the indexes of automata $\slaveA_1,  \ldots, \slaveA_l$), 
(2)~$f$ is a value function on infinite words, called the \emph{master value function}, and
(3)~$\slaveA_1, \ldots, \slaveA_l$ are weighted automata over finite words called \emph{slave automata}.
Intuitively, an NWA can be regarded as an $f$-automaton whose weights are dynamically computed at every step by a corresponding slave automaton.
We define an \emph{$(f;g)$-automaton} as an NWA where the master value function is $f$ and all slave automata are $g$-automata.

\Sskip
\Paragraph{Semantics: runs and values}.
A \emph{run} of $\nestedA$ on an infinite word $w$ is an infinite sequence 
$(\masterRun, \slaveRun_1, \slaveRun_2, \ldots)$ such that 
(1)~$\masterRun$ is a run of $\masterA$ on $w$;
(2)~for every $i>0$ we have $\slaveRun_i$ is a run of the automaton $\slaveA_{\cost(\masterRun[i-1], w[i], \masterRun[i])}$,
referenced by the label $\cost(\masterRun[i-1], w[i], \masterRun[i])$ of the master automaton, on some finite word of $w[i,j]$.
The run $(\masterRun, \slaveRun_1, \slaveRun_2, \ldots)$ is \emph{accepting} if all 
runs $\masterRun, \slaveRun_1,  \slaveRun_2, \ldots$ are accepting (i.e., $\masterRun$ satisfies its acceptance 
condition and each $\slaveRun_1,\slaveRun_2, \ldots$ ends in an accepting state)
and infinitely many runs of slave automata have length greater than $1$ (the master automaton takes infinitely many non-silent transitions).
The value of the run $(\masterRun, \slaveRun_1, \slaveRun_2, \ldots)$ is defined as 
$\silent{f}( v(\pi_1) v(\pi_2) \ldots)$, where $v(\pi_i)$ is the value of the run $\pi_i$ in 
the corresponding slave automaton.
The value of a word $w$ assigned by the automaton $\nestedA$, denoted by  
$\valueL{\nestedA}(w)$, is the infimum of the set of values of all {\em accepting} runs.
We require accepting runs to contain infinitely many non-silent transitions because
$f$ is a value function over infinite sequences, so we need 
the sequence $v(\pi_1) v(\pi_2) \ldots$ with $\bot$ symbols removed to be infinite.

\Paragraph{Deterministic nested weighted automata}. An NWA $\nestedA$ is \emph{deterministic} if (1)~the master automaton 
and all slave automata are deterministic, and (2)~slave automata recognize prefix-free languages, i.e.,
languages $\lang$ such that if $w \in \lang$, then no proper extension of $w$ belongs to $\lang$.
Condition (2) implies that no accepting run of a slave automaton visits an accepting state twice.
Intuitively, slave automata have to accept the first time they encounter an accepting state as 
they will not see an accepting state again.


\begin{definition}[Width of NWA] 
An NWA has \emph{width} $k$ if and only if
in every  run at every 
position at most $k$ slave automata are active.
\end{definition}

\begin{example}[Non-overlapping ART]
\label{ex:ARTOne}
Consider the NWA $\nestedA$ from Example~\ref{ex:ART} depicted in Fig.~\ref{fig:ART},
which does not have bounded width. The run in Fig.~\ref{fig:ART} has width at 
least $4$, but on word $rgr^2gr^3g \ldots$ the number of active slave automata 
at position of letter $g$ in subword $r^i g$ is $i$. 
We consider a variant of the ART property, called the $1$-ART property, where
after a request till it is granted additional requests are not considered.
Formally, we consider the ART property over the language $\lang_{1}$ defined by 
$(r\blank^*g\blank^*)^{\omega}$ (equivalently, given a request, the automata can check
if the slave automaton is not active, and only then invoke it).
An NWA $\nestedA_1$ computing the ART property over $\lang_{1}$ is obtained from the NWA 
from Fig.~\ref{fig:ART} by taking the product of the master automaton $\masterA$ (from Fig.~\ref{fig:ART})
with an automaton recognizing the language $\lang_{1}$. 
The automaton $\nestedA_1$, as well as, $\nestedA$ from Example~\ref{ex:ART} are 
$(\flimavg;\fsum^+)$-automata and they are deterministic.
Indeed, the master automaton and the slave automata of $\nestedA_1$ (resp., $\nestedA$) are deterministic 
and the slave automata recognize prefix-free languages.
Moreover, in any (infinite) run of $\nestedA_1$ at most one slave automaton is active, i.e., 
$\nestedA_1$ has width $1$. 
The dummy slave automata do not increase the width as they immediately accept, 
and hence they are not considered as active even at the position they are invoked.
Finally, observe that the $1$-ART property can return unbounded values, which implies that 
there exists no (non-nested) $\flimavg$-automaton expressing it.
\end{example}

\Paragraph{Decision problems}. The classical questions in automata theory are language \emph{emptiness} and \emph{universality}.
These problems have their counterparts in the quantitative setting of weighted automata and NWA.
The (quantitative) emptiness and universality problems are defined in the same way for weighted automata and NWA; 
in the following definition the automaton $\aut$ can be either a weighted automaton or an NWA.
\begin{itemize}
\item \textbf{Emptiness}: Given an automaton $\aut$ and a threshold $\lambda$, decide whether there exists a word $w$ with
$\lang_\aut(w) \leq \lambda$.
\item \textbf{Universality}: Given an automaton $\aut$ and a threshold $\lambda$, decide whether for every word $w$ we have 
$\lang_\aut(w) \leq \lambda$.
\end{itemize}
The universality question asks for \emph{non-existence} of a word $w$ such that $\lang_\aut(w) > \lambda$.

\begin{remark}
In this work we focus on value functions $\fsum$ and $\fsum^+$ for finite words, 
and $\flimavg$ for infinite words.
There are other value functions for finite words, such as $\fmax,\fmin$ and bounded
sum.  
However, it was shown in~\cite{nested} that for these value functions, there is a reduction to
non-nested weighted automata.
Also for infinite words, there are other value functions such as $\fsup,\flimsup$,
where the complexity and decidability results have been established in~\cite{nested}.
Hence in this work we focus on the most conceptually interesting case of 
$\flimavg$ function for master automaton and the $\fsum$ and $\fsum^+$ value functions
for the slave automata.
\end{remark}

%% file: examples.tex
In this section we present several examples of properties of interest that can be
specified with NWA of bounded width.

\begin{example}[Variants of ART ]
\label{ex:ARTtypes}
Recall the ART property~(Example~\ref{fig:ART})~and its variant $1$-ART property~(Example~\ref{ex:ARTOne}).
We present two variants of the ART property. 

First, we extend Example~\ref{ex:ARTOne} and consider the $k$-ART property over languages 
$L_{k}$ defined by $(\blank^*r (\blank^*r\blank^*)^{\leq k-1}  g\blank^*)^{\omega}$, i.e.,
the language where there are at most $k$-pending requests before each grant. 
As Example~\ref{ex:ARTOne}, an NWA $\nestedA_k$ computing the $k$-ART property can be constructed 
from the NWA from Fig.~\ref{fig:ART} by taking the product of the master automaton $\masterA$ 
(from Fig.~\ref{fig:ART}) with an automaton recognizing $\lang_{k}$. 
The NWA $\nestedA_k$ has width $k$.

Second, we consider the $1$-ART$[k]$ property, where $\Sigma = \{ r_i,g_i : i \in \{1,\ldots, k\}\} \cup \{\blank \}$, 
i.e., there are $k$-different types of ``request-grant'' pairs. The $1$-ART$[k]$ property
asks for the average number of instructions between any request and the following grant of the corresponding type.
Moreover, we consider as for $1$-ART property that for every $i$, between a request $r_i$ and the following grant of the corresponding type 
$g_i$, there is no request $r_i$ of the same type. 
The  $1$-ART$[k]$ can be expressed with an $(\flimavg;\fsum^+)$-automaton $\nestedA_{1}^{[k]}$ of width bounded by $k$, 
which is similar to $\nestedA_1$ from Example~\ref{ex:ARTOne}. 
Basically, the NWA $\nestedA_{1}^{[k]}$ has $k$ slave automata; for $i \in \{1,\ldots,k\}$
the slave automaton $\slaveA_i$ is invoked on letters $r_i$ and it counts 
the number of steps to the following grant $g_i$.
Additionally, the master automaton checks that for every $i$, between any two grants $g_i$, there is at most one request $r_i$. 
\end{example}

In Examples~\ref{ex:ART},~\ref{ex:ARTOne}, and~\ref{ex:ARTtypes} we presented properties 
that can be expressed with $(\flimavg;\fsum^+)$-automata. 
The following property of \emph{average excess} can be expressed with slave automata with $\fsum$ value functions 
that have both positive and negative weights, i.e., it can be expressed by an $(\flimavg;\fsum)$-automaton, 
but not $(\flimavg;\fsum^+)$-automata.

\begin{example}[Block difference]
\newcommand{\nestedAE}{\nestedA_{\textrm{AE}}}
Consider the alphabet $\{ r,g,\blank \}$ from Example~\ref{ex:ART} with an additional letter $\delim$. 
The \emph{average excess} (AE) property asks for the average difference between requests and grants 
over blocks separated by $\delim$. For example, for $\delim (r r \# g \delim)^{\omega}$ the average excess 
is~$1$.
The AE property can be expressed by $(\flimavg;\fsum)$-automaton $\nestedAE$ of width $1$ (presented below), 
but it cannot be expressed with $(\flimavg;\fsum^+)$-automata; $(\flimavg;\fsum^+)$-automata return values 
form the interval $[0,\infty)$, while AE ranges from $(-\infty, \infty)$.
The automaton $\nestedAE$ invokes a slave automaton $\slaveA_1$ at positions of letter $\delim$ 
and a dummy automaton $\slaveA_2$ on the remaining positions. 
The slave automaton $\slaveA_1$ runs until it sees $\delim$ letter; 
it computes the difference between $r$ and $g$ letters by taking transitions of weights 
$1, -1, 0$ respectively on letters $r,g,\blank$.
The master automaton as well as the slave automata of $\nestedAE$ are deterministic and 
the slave automata recognize prefix-free languages.
Therefore, the NWA $\nestedAE$ is deterministic and has width~1.
\end{example}

%% file: results.tex
\newcommand{\iterWord}{u}
\newcommand{\optWord}{w}
\newcommand{\conf}{\textsc{Conf}}

In this section we establish our main results. 
We first discuss complexity of checking whether a given NWA has width $k$. 
Next, we comment the results we need to prove. Afterwards, we present our results.

\Paragraph{Configurations}. 
Let $\nestedA$ be a non-deterministic $(\flimavg; \fsum)$-automaton of width $k$.
We define a \emph{configuration} of $\nestedA$ as a tuple 
$(q; q_1, \ldots, q_k)$ where $q$ is a state of the master automaton and 
each $q_1, \ldots, q_k$ is either a state of a slave automaton of $\nestedA$ 
or $\bot$. 
In the sequence $q_1, \ldots, q_k$ each state corresponds to one slave 
automaton, and the states are ordered w.r.t.\ the position when the 
corresponding slave automaton has been invoked, i.e., $q_1$ correspond to the 
least recently invoked slave automaton. 
If there are less than $k$ slave automata active, then $\bot$ symbols follow 
the actual states (denoting there is no slave automata invoked).
We define $\conf(\nestedA)$ as the number of configurations of $\nestedA$.

\noindent \emph{Key ideas}. 
NWA without weights are equivalent to \buchi{} automata~\cite{nested}.
The property of having width $k$ is independent from weights. It can be decided with 
a constriction of a (non-weighted) \buchi{} automaton, which tracks configurations 
$(q; q_1, \ldots, q_k)$ of a given NWA (assuming that is has width $k$) and accepts only if the 
width-$k$ condition is at some point violated.

\begin{restatable}{theorem}{checkingWidth}
\label{th:checkingWidth}
(1)~Fix $k>0$. We can check in polynomial time whether a given NWA has width $k$.
(2)~Given an NWA and a number $k$ given in unary we can check in polynomial space whether the  NWA has width $k$.
\end{restatable}
\begin{proof}
Let $\nestedA$ be an NWA and let $k>0$ be a tested width. 
Consider a \buchi{} automaton $\nonnestedA$, whose states are configurations of $\nestedA$ (considered to have width $k$) 
and a single accepting state $q_{acc}$.
The automaton $\nonnestedA$ simulates runs of $\nestedA$, i.e., it has a transition from one configuration to another
over letter $a$ if and only if there exist corresponding transitions of the master automaton and slave automata over letter $a$, which result in such a transition of $\nestedA$. This condition can be checked in polynomial time;
for configurations $\tuple{q;q_1, \ldots, q_k}, \tuple{q';q_1', \ldots, q_k'}$ and letter $a$,
we need to check whether $\tuple{q,a,q'}$ is a transition of the master automaton of $\nestedA$
and each transition $\tuple{q_1,a,q_1'}, \ldots, \tuple{q_k,a,q_k'}$ is a transition of (some) slave automaton of $\nestedA$.
Additionally, whenever $\nonnestedA$ is in a state $\tuple{q; q_1, \ldots, q_k}$ such that 
$q_k \neq \bot$, i.e., $k$ slave automata are active, and another slave automaton is invoked, 
$\nonnestedA$ takes a transition to $q_{acc}$, which is a single accepting state in $\nonnestedA$.
Observe that $\nonnestedA$ has an accepting run if and only if $\nestedA$ violates width-$k$ condition. 

The size of $\nonnestedA$ is bounded by $|\nestedA|^k$, i.e., it is 
polynomial in the size of $\nestedA$ and exponential in $k$. 
Therefore, 
we can check emptiness of $\nonnestedA$, and in turn violation of width-$k$ condition, in polynomial time if $k$ is constant 
and $\PSPACE$ if $k$ is given in input in unary.
\end{proof}

\smallskip\noindent{\em Comment.}
We first note that for deterministic automata, emptiness and universality 
questions are similar. 
Hence we focus on the emptiness problem for non-deterministic automata
(which subsumes the emptiness problem for deterministic automata) to 
establish the new results of Table~\ref{tab1}. 
Moreover, the $\fsum^+$ value function is a special case of the $\fsum$ value 
function with only positive weights.
Since our main results are algorithms to establish upper bounds, we will only 
present the result for the emptiness problem for non-deterministic 
$(\flimavg;\fsum)$-automata. 
However, as a first step we show that without loss of generality, we
can focus on the case of deterministic automata.


\begin{restatable}{lemmaStatement}{ExmpinessDeterministic}
Let $k > 0$.
Given a non-deterministic $(\flimavg; \fsum)$-automaton $\nestedA$ over 
alphabet $\Sigma$ of width $k$, a deterministic 
$(\flimavg; \fsum)$-automaton $\nestedA_d$ of width $k$ over an alphabet 
$\Sigma \times \Gamma$ such that $\inf_{\iterWord \in \Sigma^{\omega}} \nestedA(\iterWord) 
= \inf_{\iterWord' \in (\Sigma \times \Gamma)^{\omega}} \nestedA_d(\iterWord')$
can be constructed in time exponential in $k$ and polynomial in $|\nestedA|$.
Moreover, $\conf(\nestedA_d)$ is polynomial in $\conf(\nestedA)$ and $k$ and only the alphabet of $\nestedA_d$ is exponential (in $k$)
as compared to the alphabet of $\nestedA$.
\label{l:deterministic}
\end{restatable}
\ExmpinessDeterministic*
\begin{proof}
The general idea is to extend the alphabet to encode the 
actual letter and auxiliary symbols, which indicate how to
resolve non-determinism.
This can be done by explicitly 
writing down transitions the master and active slave automata should take. 
However, in such a solution, two slave automata which are 
at the same position in the same state, have the same suffixes of their runs, even though
in a non-deterministic automaton their (suffixes) of runs can be different.

To circumvent this problem, we modify the automaton $\nestedA$ to $\nestedA'$ such that 
each slave automaton comes in $k$ copies and the master automaton of $\nestedA'$ can 
invoke any copy of a slave automaton, i.e., if the master automaton of $\nestedA$ has a transition
$(q,a,q',i)$, at which it invokes the slave automaton $\slaveA_i$, 
the master automaton of $\nestedA$ can invoke any copy of $\slaveA_i$.
Clearly, $\nestedA'$ does not have any additional behaviors, i.e., by 
merging copies of slave automata in a run of $\nestedA'$,
 we can obtain a run $\nestedA$. Conversely, 
for every run of $\nestedA$ there exist multiple corresponding runs 
of $\nestedA'$. In particular, for every run of $\nestedA$ there exists
a run of $\nestedA'$ such that at every position, all active slave automata are different.
We call such runs \emph{controllable} as every slave automaton can be controlled independently of the others. 
In the following, we construct a deterministic automaton $\nestedA_d$, which has corresponding run to every controllable run of $\nestedA'$.

Before we describe the construction of $\nestedA_d$, observe that without loss of generality, we can assume that accepting states of
slave automata do not have outgoing transitions. Intuitively, we can clone each accepting state $s$ into two copies $s_1, s_2$, of which 
$s_1$ is a state with the transitions of $s$, but is not accepting, and $s_2$ is accepting but has no outgoing transitions. 
Then, all transitions to $s$ are changed into two transitions, one to $s_1$ and one to $s_2$.

\newcommand{\Qq}{\bm{Q}}

Let $\Qq$ be the union of the set of states of the master automaton and all sets of states of slave automata of $\nestedA'$.
We define $\Gamma$ as the set of partial functions $h$ 
from $(k+1)$-elements subsets of $\Qq$ into $\Qq$.
We define an $(\flimavg; \fsum^+)$-automaton $\nestedA_d$ over the alphabet $\Sigma \times \Gamma$ 
by modifying only the transition relations and labeling functions of the master automaton and slave automata of $\nestedA'$; 
the sets of states and accepting states are the same as in the original automata.
The transition relation and the labeling function of the master automaton $\masterA^d$ of
$\nestedA_d$ is defined as follows: for all states $q,q'$, 
$(q, \lpair{a}{h},q')$ iff $h(q) = q'$ and the master automaton of $\nestedA'$ has the transition $(q,a,q')$.
The label of the transition $(q, \lpair{a}{h},q')$ is the same as the label 
of the transition $(q,a,q')$.  
Similarly,  for each slave automaton $\slaveA_i$ in $\nestedA'$, the transition relation of the corresponding
slave automaton $\slaveA_i^d$ in $\nestedA_d$ is defined as follows:
for all states $q,q'$ of $\slaveA_i^d$, 
$(q, \lpair{a}{h},q')$ is a transition of $\slaveA_i^d$
 iff 
there $h(q) = q'$ and $\slaveA_i$ has the transition $(q,a,q')$.
The label of the transition $(q, \lpair{a}{h},q')$ is the same as the label 
of the transition $(q,a,q')$ in $\slaveA_i^d$. 
Observe that $(q, \lpair{a}{h},q')$ can be a transition of $\slaveA_i^d$ only if $h(q)$ is defined.

Observe that the master automaton of $\nestedA_d$ and all slave automata $\slaveA_i^d$
are deterministic. Moreover, since we assumed that for every slave automaton in $\nestedA'$
final states have no outgoing transitions, slave automata $\slaveA_i'$ recognize 
prefix free languages. 
Finally, it follows from the construction that
(i)~for every controllable run of $\nestedA'$, there exists a corresponding run of $\nestedA_d$ where the sequence of values returned by slave automata
is the same as in the run of $\nestedA'$; such runs have the same value.
Basically, we encode in the input word, transitions of all automata with functions $h \in \Gamma$.
Due to controllability of the run, at every position every slave automaton is in a different state.
Therefore, we can encode in $h$ transitions of all slave automata as well as the master automaton.
Conversely, (ii) a run of $\nestedA_d$, which is as a sequence of sequences of states, is basically a run of $\nestedA$ and it clearly has the same value. 
Therefore, the infimum over all words of $\nestedA_d$ coincides with the infimum over all words of $\nestedA'$ as well as of $\nestedA$.
\end{proof}


\smallskip\noindent{\bf Proof overview.} 
We present our proof overview for the emptiness of deterministic 
$(\flimavg; \fsum)$-automata. 
The proof consists of the following four key steps.
\begin{compactenum}
\item First, we identify a condition, and show in Lemma~\ref{l:Infinity} that 
it is a sufficient condition to ensure that the infimum value among all words 
is $-\infty$ (i.e., the least value possible).
Moreover we show that the condition can be decided in 
$\PTIME$ if $k$ is 
constant (even $\NLOGSPACE$ if additionally the weights are in unary)
and in $\PSPACE$ if $k$ is given in unary.

\item Second, we show that if the above condition does not hold, then there is 
a family of lasso words (i.e., a finite prefix followed by an infinite 
repetition of another finite word) that approximates 
the infimum value among all words. 
This shows that the above condition is both necessary and sufficient.
Moreover, we consider {\em dense} words, where if we consider that slave automata have been
invoked for the $i$-th time, then the run of the slave automata invoked is at most for 
$O(\log (i))$ steps. 
We show that the infimum is achieved by a dense word. 
These results are established in Lemma~\ref{l:dense}.

\item Third, we show using the above result, that the problem for bounded width
can be reduced to the problem of width~1, and the reduction is polynomial in the size
of the original automaton, and only exponential in $k$. Thus if $k$ is constant, the
reduction is polynomial.
This is established in Lemma~\ref{l:reduction}.

\item Finally, we show that for automata with width~1, the emptiness problem can be 
solved in $\NLOGSPACE$ if weights are in unary and otherwise in 
$\PTIME$ (Lemma~\ref{l:WidthOnePoly}).
\end{compactenum}
Given the above four steps we conclude our main result (Theorem~\ref{th:main}).
We start with the first item.

\smallskip\noindent{\em Intuition for the condition.} 
We first illustrate with an example that for very similar automata, which
just differ in order of invoking slave automata, the infimum over the values
are very different. 
For one automaton the infimum value is~$-\infty$ and for the other it is~0.
This example provides the intuition for the need of the condition to identify
when the infimum value is~$-\infty$.

\begin{example}
\label{ex:condMotivation}
Consider two deterministic $(\flimavg;\fsum)$-automata $\nestedA_1, \nestedA_2$ defined as follows.
The master automaton $\masterA$ of $\nestedA_1$ accepts the language $(12a^*\#)^{\omega}$.
At letter $1$ (resp., $2$) it invokes an automaton $\slaveA_1$ (resp., $\slaveA_2$).
The slave automaton $\slaveA_1$ increments its value at every $a$ letter and it terminates once it reads $\#$.
The slave automaton $\slaveA_2$ works as $\slaveA_1$ except that it decrements its value at $a$ letters.
NWA $\nestedA_2$ is similar to $\nestedA_1$ except that it accepts the language $(21a^*\#)^{\omega}$.
It invokes the same slave automata as $\nestedA_1$.
Thus the two automata only differ in the order of invocation of the slave automata.
Observe that the infimum over values of all words in $\nestedA_1$ is $0$. 
Basically, the values of slave automata are always the opposite, therefore
the average of the values of slave automata is $0$ infinitely often.
However, the infimum over values of all words in $\nestedA_2$ is $-\infty$. Indeed, consider 
a word $21a^1\# \ldots 21a^{2^i} \ldots$. At positions proceeding $1a^{2^i}$, the automaton 
$\slaveA_2$ returns the value $-2^i$ and the average of all previous $2 \cdot i$ values is $0$. 
Thus, the average at this position equals $-\frac{2^i}{2\cdot i}$ (recall that the average
is over the number of invocations of slave automata). 
Hence, the limit infimum of averages is $-\infty$.
\end{example}

\Paragraph{Condition for infinite infimum.}
Let $k > 0$ and $\nestedA$ be a deterministic $(\flimavg;\fsum)$-automaton of width $k$.
Let $C$ be the minimal weight of slave automata of $\nestedA$.
Condition~(*):
\begin{description}
\item[(*)] $C<0$ and there exists a word $w$ accepted by $\nestedA$ and infinitely many positions $b$ such that 
the sum of weights, which automata active at position $b$ accumulate while running on $w[b, \infty]$, is smaller than $C \cdot k^2 \cdot \conf(\nestedA)$.
\end{description}
Intuitively, condition (*) implies that there is a subword $u$ which can be repeated so that the values of slave automata 
invoked before position $b$ can be decreased arbitrarily. Note that pumping that word may not decrease the total
average of the word. However, with $\flimavg$ value function, we need to ensure only the existence of a subsequence of positions 
at which the averages tend to $-\infty$, i.e., we only need to decrease the 
values of slave automata invoked before position $b$ (for infinitely many 
positions).

\smallskip\noindent{\em Illustration of condition on example.}
Consider automata $\nestedA_1, \nestedA_2$ from Example~\ref{ex:condMotivation}.
The automaton $\nestedA_2$ satisfies condition (*), whereas $\nestedA_1$ does not. 
In the word $21a^1\# \ldots 21a^{2^i} \ldots$, consider positions $b$, where $\slaveA_2$ is invoked by $\nestedA_2$. 
The automaton $\slaveA_2$ works on the subword $21 a^{2^i}$, where both automata $\slaveA_1, \slaveA_2$ are active and the sum of their values
past any position is $0$. However, the only slave automaton active at position $b$ is $\slaveA_2$. 
These automaton accumulates the value $-2^i$ past position $b$. Therefore, past some position $N$, 
all such positions $b$ satisfy the statement from condition (*), and hence $\nestedA_2$ satisfies condition (*).
Now, for  $\nestedA_1$, at every position at which $\slaveA_2$ is active, $\slaveA_1$ is active as well, hence
for any position $b$, the values accumulated by slave automaton active past this position is non-negative.
Hence, $\nestedA_1$ does not satisfy condition (*).
We now present our lemma about the condition.

\begin{restatable}{lemmaStatement}{Condition}
Let $k > 0$ and $\nestedA$ be a deterministic $(\flimavg;\fsum)$-automaton of width $k$. 
\begin{compactenum}[(1)]
\item If condition (*) holds for $\nestedA$,  then $\inf_{\iterWord \in \Sigma^{\omega}} \nestedA(\iterWord) = -\infty$.
\item Condition (*) can be checked in $\NLOGSPACE$ if the width is constant and weights are given in unary,
$\PTIME$ if the width is constant,
 and in $\PSPACE$ if the width is given in unary. 
\end{compactenum}
\label{l:Infinity}
\end{restatable}

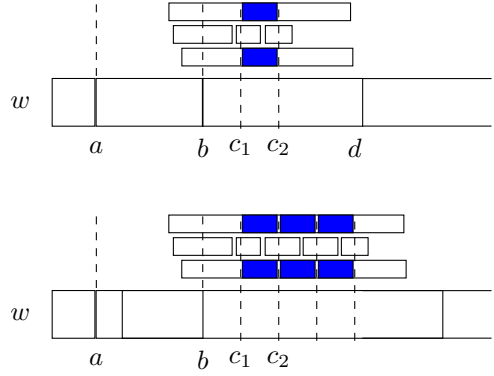
\begin{figure}
\centering
\begin{tikzpicture}
\pumpingOne 
\end{tikzpicture}
\caption{Explanation to the proof in Lemma~\ref{l:Infinity}}
\label{fig:pumpingOne}
\end{figure}

\begin{proof} We present proofs for each item below.

\Paragraph{Proof of (1)}.
\newcommand{\pos}[1]{\textbf{#1}}
Assume that (*) holds. 
We show that there exists a word $u'$ such that $\nestedA(u') = -\infty$. 
Consider a word $w$ and position $\pos{b}$ which 
(1)~satisfy  condition (*), and 
(2)~the configuration at the position $\pos{a}$, the position of invocation of least recent slave automaton active at $\pos{b}$, 
occurs infinitely often in the run of $\nestedA$ in $w$.

At every step all slave automata decrease the sum of weights by at most $C \cdot k$. 
Therefore, there exist more than $k \cdot \conf(\nestedA)$ positions in $w[\pos{b},\infty]$ at which the 
the sum of weights of automata invoked before position $\pos{b}$ decreases.
Thus there exist positions $\pos{x}_1, \pos{x}_2$
between which no automaton invoked before position $\pos{b}$ terminates and there are more than 
$\conf(\nestedA)$ positions at which the sum of values of these automata decreases.  

Therefore, there exist positions $\pos{c}_1, \pos{c}_2$ such that the configurations at $\pos{c}_1, \pos{c}_2$ are the same, and 
the sum of values of all automata activated before $\pos{b}$, which are still active at $\pos{c}_2$ decreases between $\pos{c}_1$ and $\pos{c}_2$ (see Fig.~\ref{fig:pumpingOne}).
Let $\pos{d}$ be a position past $\pos{c}_2$ with the same configuration as at $\pos{a}$
 such that the master automaton visits an accepting state between $\pos{a}$ and $\pos{d}$.
Recall that all slave automata active at $\pos{b}$ have been invoked past $\pos{a}$.
We show how to build the word $u'$ form subwords $w[1, \pos{a}], w[\pos{a}, \pos{c}_1], w[\pos{c}_1, \pos{c}_2]$ and $w[\pos{c}_2, \pos{d}]$.

Let $f(i) = 2^{2^i}$, i.e., $f(n+1) = f(n)\cdot f(n)$.
Consider the word $u' = w[1,\pos{a}] \gamma_1 \gamma_2 \ldots$, where
$\gamma_i = w[\pos{a},\pos{c}_1] (w[\pos{c}_1,\pos{c}_2])^{f(i)} w[\pos{c}_2, \pos{d}]$. 
Observe that the sum of values of slave automata invoked in the prefix $w[1,\pos{a}] \gamma_1 \ldots \gamma_n w[\pos{a},\pos{c}_1]$, which is shorter
than $2 \cdot f(n)$ is less than $k \cdot C \cdot 2\cdot f(n) - f(n+1) < - f(n) (f(n) - 2\cdot k \cdot C)$. 
At most $2\cdot f(n)$ slave automata have been invoked in this prefix, hence the average is less than $0.5 \cdot (-f(n) + k \cdot C)$. 
Hence, the limit infimum tends to minus infinity. Moreover, the master automaton visits one of its accepting states at least once in each $\gamma_i$.

\Paragraph{Proof of (2)}.
\newcommand{\config}{\textsc{Cnf}}
Consider a graph $G(\nestedA)$ of configurations of $\nestedA$, in which there exists an edge from configuration $\config_1$ to $\config_2$ if and only if 
the automaton $\nestedA$ has a transition from $\config_1$ to $\config_2$. 
Recall that in a configuration $(q;q_1, \ldots q_j, q_{j+1}, \ldots, q_{k})$ the states $q_1, \ldots, q_j$
correspond to the $j$ least recently invoked slave automata.
We show that condition (*) holds if and only if (**)~there exists a cycle in $G(\nestedA)$ such that 
\begin{compactenum}[(a)]
\item a configuration $\config$ with an accepting state of the master automaton of $\nestedA$ is reachable from the cycle and the cycle is reachable from $\config$, and
\item for some $j\geq 1$, the sum of weights of the $j$ least recently invoked slave automata is negative in this cycle.
\end{compactenum}

\noindent $\bm{(*) \Rightarrow (**)}$: Consider word $w[1,\pos{a}] w[\pos{a},\pos{c}_1] w[\pos{c}_1, \pos{c}_2]$. 
The configurations (in the run of $\nestedA$) along the subword $w[\pos{c}_1,\pos{c}_2]$ define a cycle satisfying  (a) and (b).
Indeed, the master automaton of $\nestedA$ visits its accepting state on $w[1,\pos{a}] w[\pos{a},\pos{c}_1] w[\pos{c}_1,\pos{c}_2] w[\pos{c}_2, \pos{d}]$; hence (a) holds.
Let $j$ be the number of slave automata invoked at position $\pos{a}$.  Then, the sum of weights of the $j$ least recently invoked slave automata on $w[\pos{c}_1, \pos{c}_2]$
is negative; hence (b) holds.

\noindent $\bm{(**)\Rightarrow (*)}$: Let $D$ be the maximal absolute weight in nested automata of $\nestedA$.
Consider a word which corresponds to a path in $G(\nestedA)$ from the initial configuration to the cycle satisfying  $(**)$ and repeating infinitely the following
extended cycle: looping at the cycle from (**) $2 \cdot D \cdot k^2 \cdot \conf(\nestedA)$ times, and finally going through all configurations to terminate all slave automata,
visit an accepting state of $\nestedA$ and returning to the start of the extended cycle.
Observe that terminating all slave automata can be done in at most $k \cdot \conf(\nestedA)$ steps, therefore all slave automata active at the beginning 
of the extended cycle accumulate in total the weight smaller than $-2 \cdot D \cdot k^2 \cdot \conf(\nestedA)$ (iterating the negative cycle) 
plus $D \cdot k \cdot \conf(\nestedA)$ (terminating slave automata), which is smaller than $C \cdot k^2 \cdot \conf(\nestedA)$; hence (*) holds.

Finally, checking existence of such a cycle can be done in $\NLOGSPACE$ w.r.t.\ to the size of $G(\nestedA)$. If the width $k$ is constant, then 
the size $G(\nestedA)$ is $O(|\nestedA|^k)$, hence it is polynomial (in $\nestedA$) if $k$ is constant and exponential if $k$ is given in unary.
Thus, checking condition (*) is $\NLOGSPACE$ for constant width and $\PSPACE$ if the width is given in unary.
\end{proof}


\begin{definition}
Let $\nestedA$ be a deterministic $(\flimavg;\fsum)$-automaton of width $k$.
A word $w$ is \emph{dense} (w.r.t.\ $\nestedA$) if in the run of $\nestedA$ on $w$, 
for every $i>0$, the $i$-th invoked slave automaton takes at most $O(\log(i))$ steps.
\end{definition}

\noindent{\em Intuitive explanation of dense words.} 
In a deterministic $(\flimavg;\fsum)$-automaton, the average is over the number of invoked slave automata, but 
in general, the returned values of the slave automata can be arbitrarily large
as compared to the number of invocations, and hence the partial averages need 
not converge.
Intuitively, in dense words, slave automata are invoked and terminated 
relatively densely, i.e., the length of their run depends on the number of 
slave automata invoked till this position.
In consequence, the value they can accumulate is small w.r.t.\ the average, 
i.e.,
their absolute contribution to the sum of first $n$ elements is $O(\log(n))$, and 
hence the contribution of the value a single slave automaton converges to $0$ 
and partial averages converge on dense words. 

\smallskip\noindent{\em Illustration on example.}
Consider an automaton $\nestedA_1$ from Example~\ref{ex:condMotivation}. We discuss the definition of density on an example of 
word $w = 12a^{1}\#12a^{3}\# \ldots 12a^{2\cdot i +1} \#\ldots $, which is not dense (w.r.t.\ $\nestedA_1$). 
Observe that at the position of subword $12a^{2\cdot i +1}$, the partial average is~$0$.
Once $\slaveA_1$ is invoked it returns value ${2\cdot i +1}$ and it is $({2\cdot i +1})$-th 
invocation of a slave automaton. 
Hence, the average increases to $1$ only to be decreased to $0$ after invocation of $\slaveA_2$. 
However, the word $w' = 12a^{1} \#(12a^{2}\#)^3 \ldots (12a^{2\cdot i +1} \#)^{2^i} \ldots$ is dense.
Indeed, before the slave automata invoked at subword $12a^{2 \cdot i+1} \#$ there are at least $\sum_{j=1}^{i-1} 2^j = 2^{i}-1$ 
invoked slave automata.
Therefore, the value ${2\cdot i +1}$ returned by $\slaveA_1$ invoked on $12a^{2\cdot i +1}$
changes the average by at most $\frac{2\cdot i +1}{2^i}$; as previously invoking $\slaveA_2$ in the next 
step bring the average back to $0$.
Therefore, the sequence of partial averages of values returned by slave automata converges to $0$.
\begin{figure}
\centering
\begin{tikzpicture}
\pumpingTwo
\end{tikzpicture}
\caption{Explanation to Lemma~\ref{l:dense}; the blue part corresponds to $H$, while the green part corresponds to $T$}
\label{fig:clarify}
\end{figure}
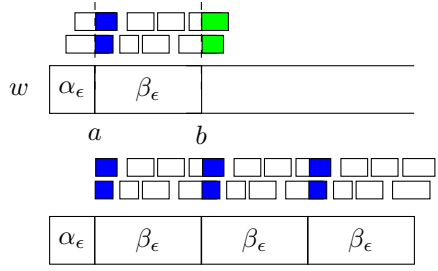

\begin{restatable}{lemmaStatement}{Density}
Let $k > 0$ and $\nestedA$ be a deterministic $(\flimavg;\fsum)$-automaton of width $k$.
Assume that condition (*) does not hold. Then the following assertions hold:
\begin{compactenum}[(1)]
\item For every $\epsilon >0$ there exist finite words $\alpha_{\epsilon}, \beta_{\epsilon}$
such that $|\inf_{\iterWord \in \Sigma^\omega} \nestedA(\iterWord)- \nestedA(\alpha_{\epsilon} (\beta_{\epsilon})^{\omega}) | < \epsilon$.
\item The value $\inf_{\iterWord \in \Sigma^{\omega}} \nestedA(\iterWord)$ is greater than $-\infty$. 
\item There exists a dense word $\optWord_d$ such that $\inf_{\iterWord \in \Sigma^\omega} \nestedA(\iterWord) = 
\nestedA(\optWord_d)$.
\end{compactenum}
\label{l:dense}
\end{restatable}
\begin{proof} We present the proof of each item.

\Paragraph{Proof of (1)}:
\newcommand{\optWordEps}{\optWord_{\epsilon}}
Consider $\epsilon >0$.
Let $\optWord_{\epsilon}$ be a word such that $\nestedA(\optWordEps) - \inf_{\iterWord \in \Sigma^{\omega}} \nestedA(\iterWord) < \frac{\epsilon}{4}$.
There exists $N_0$ such that past $N_0$, all configurations of $\nestedA$ occur infinitely often in the run of $\nestedA$ on $\optWordEps$ and
no position $b > N_0$ in $\optWord$ satisfies condition (*).
Let $N_0 < a < b$ be positions in $\optWordEps$ such that 
\begin{compactenum}[(1)]
\item the configuration of $\nestedA$ at $a$ and $b$ is the same,
\item the average of the values of slave automata invoked between $a$ and $b$ is at most $\nestedA(\optWordEps) + \frac{\epsilon}{4}$,
\item the number $n$, which is the number of slave automata invoked between positions $a$ and $b$, is greater that 
$\frac{4}{\epsilon} \cdot C \cdot k^2 \cdot \conf(A)$, where $C$ is the minimal weight, and
\item the sum of the weights slave automata active at $a$ accumulate past $a$ is at most $\frac{n \cdot \epsilon}{4}$,
\item all slave automata invoked before $a$ are terminated before $b$,

\item the master automaton of $\nestedA$ visits between $a$ and $b$ an accepting state at least once.
\end{compactenum}
Let $P > N_0$ be the minimal position such that all configurations, which appear infinitely often the run of $\nestedA$ on $\optWordEps$, 
appear between $N_0$ and $P$. Let $M$ be the maximal absolute sum of slave automata invoked before $P$. 
We pick $b$ such that the average of all slave automata invoked before $b$ is at most  $\nestedA(\optWordEps) + \frac{\epsilon}{8}$,
and $b$ is large enough to satisfy (3), (6) and (4) and (5) for $a = P$. For such $b$ we choose $a$ from the interval $[N_0, P]$ and observe that (1)--(6) are satisfied.
 
We put $\alpha_{\epsilon} = u[1,a]$ and 
$\beta_{\epsilon} = u[a,b]$. We claim that  $|\nestedA(\optWordEps) - \nestedA(\alpha_{\epsilon} (\beta_{\epsilon})^{\omega}) | < \epsilon$.
Clearly, $\nestedA(\optWordEps) < \nestedA(\alpha_{\epsilon} (\beta_{\epsilon})^{\omega})$. 
Conversely, due to condition (3), we have 
$\nestedA(\alpha_{\epsilon} (\beta_{\epsilon})^{\omega}) < \nestedA(\optWordEps) + \frac{\epsilon}{4} - T + H$, where
$T$, the tail, is the sum of weights, which automata invoked before position $b$ accumulate past position $b$, and 
$H$, the head, is the sum of weights, which automata invoked before position $b$ accumulate while running on $\beta_{\epsilon}$ again (see Fig~\ref{fig:clarify}).
By the condition (5), we can estimate $H$ by $\frac{n \cdot \epsilon}{4}$. 
Due to (4) and by condition (*),  we have $T < C \cdot k^2 \cdot \conf(\nestedA)$. 

\Paragraph{Proof of (2)}:
Assume that condition (*) does not hold. Then, for $\epsilon = 1$ there exist $\alpha_{\epsilon}, \beta_{\epsilon}$ such that 
$|\nestedA(\alpha_{\epsilon} (\beta_{\epsilon})^{\omega}) -\inf_{\iterWord \in \Sigma^{\omega}} \nestedA(\iterWord) | < 1$. However, observe that 
$\nestedA(\alpha_{\epsilon} (\beta_{\epsilon})^{\omega})$ is finite. Indeed, all slave automata invoked in $\beta_{\epsilon}$ has to terminate within
$|\beta_{\epsilon}|$ steps, otherwise there is a slave automaton with an infinite run. It follows that the value of 
$\nestedA(\alpha_{\epsilon} (\beta_{\epsilon})^{\omega})$ is greater of equal to $C \times k \times |\beta_{\epsilon}|$, where $C$ is the minimal weight in 
all slave automata of $\nestedA$. 
In consequence, $\inf_{\iterWord \in \Sigma^{\omega}} \nestedA(\iterWord) > -\infty$.

\Paragraph{Proof of (3)}:
We can strengthen~(1) and say that
there exist pairs of words $(v_1, u_1), (v_2, u_2), \ldots$ such that
(1)~for every $i$, we have $|\nestedA(v_i u_i^{\omega}) - \nestedA(\optWordEps)| < \frac{1}{i}$, and
(2)~for all $i,j$, the automaton $\nestedA$ reaches the same configuration from the initial state upon
reading $v_i$ and $v_j$.
Having conditions (1) and (2) we define word $\optWord$  as $\alpha_1 \beta_1^{k[1]} \beta_2^{k[2]} \ldots$,
where the sequence $k[i]$ is defined as follows:
$k[1] = 1$ and for $i\geq 1$ we have $k[i] = 2^{|\beta_{i-1}| + |\beta_i| + |\beta_{i+1}|}$.
Clearly, $\inf_{\iterWord \in \Sigma^{\omega}} \nestedA(\iterWord) = \nestedA(\optWord')$. 
Observe that every slave automaton invoked at a subword $\beta_i$ terminates within $max(|\beta_i|, |\beta_{i+1}|)$ steps, depending on whether the following subword in 
$\beta_i$ or $\beta_{i+1}$.
However, the first subword $\beta_i$ appears first after at least $2^{|\beta_i|+|\beta_{i+1}|}$ slave automata has been invoked.
Therefore, the word $\optWord'$ is dense.
\end{proof}

\begin{remark}
Lemma~\ref{l:Infinity} together with (2) of Lemma~\ref{l:dense} imply that 
for a deterministic $(\flimavg;\fsum)$-automaton $\nestedA$ of width $k$ condition 
(*) is both necessary and sufficient for
the infimum over all values equal to $-\infty$. 
Moreover, this condition can be checked efficiently. 
\end{remark}

Lemma~\ref{l:reduction} reduces the emptiness problem for deterministic $(\flimavg;\fsum)$-automata of width $k$ 
to the same problem with automata of width $1$. 

\begin{restatable}{lemmaStatement}{Reduction}
\label{l:reduction}
Let $k > 0$ and $\nestedA$ be a deterministic $(\flimavg;\fsum)$-automaton of width $k$.
Assume that condition (*) does not hold. Then, there exists a deterministic $(\flimavg;\fsum)$-automaton $\nestedA_1$ 
of width $1$ over an alphabet $\Delta$ such that
$\inf_{\iterWord \in \Sigma^{\omega}} \nestedA(\iterWord) = \inf_{\iterWord \in \Delta^{\omega}} \nestedA_1(\iterWord)$.
The size of $\nestedA_1$ is $O(|\nestedA|^k)$ and it can be constructed on-the-fly.
\end{restatable}
\begin{proof}
We define the automaton $\nestedA_1$ of width bounded by $1$, whose slave automaton $\slaveA^{\Sigma}$ simulates runs of $k$ slave automata of $\nestedA$.
The slave automaton  $\slaveA^{\Sigma}$ computes the sum of weights collected by all active slave automata of $\nestedA$. Once $\nestedA$ invokes a new slave automaton,
$\slaveA^{\Sigma}$ terminates and it is immediately restarted in the next transition.
Moreover, if the last active slave automaton of 
$\nestedA$ terminates, $\slaveA^{\Sigma}$ terminates as well.
More formally, the master automaton of $\nestedA_1$ is $(\Sigma, Q, q_0, \delta, F,\cost)$, where
$Q = Q_m \times (Q_s \times \ldots \times Q_s) \times \{ 0,1\} = Q_m \times Q_s^k \times \{0,1\}$, where 
$Q_m$ is the set of states of the master automaton of $\nestedA$ and 
$Q_s$ is the union of the set of states of
The master automaton keeps track of the master automaton of $\nestedA$ and all its active slave automata. 
Its last bit $\{0,1\}$ indicates whether the slave automaton is active, i.e., if one of slave automata of $\nestedA$ terminates, the bit is flipped to $0$, and
if the bit is $0$ but some of slave automata of $\nestedA$ are active, then start the slave automaton in the configuration corresponding to the current configuration of slave automata.
In the generalized \buchi{} acceptance condition $F$, we encode that the run of the master automaton accepts and runs of all slave automata are finite.  
The components $q_0, \delta, F,\cost$ are defined accordingly to the description.

The slave automaton $\slaveA^{\Sigma}$ has the similar structure to the master automaton of $\nestedA_1$; there are two key difference.
First, if any of tracked slave automata of $\nestedA$ terminates, $\slaveA^{\Sigma}$ terminates as well.
Second, for every transition of $\slaveA^{\Sigma}$, the weight of this transition is the sum of weights of current transitions of tracked slave automata.

Now, we show that on dense words (w.r.t. $\nestedA$), values of both automata coincide. 
Let $\optWord$ be a dense word. 
Consider a position $i$ in $\optWord$. Let $n$ be the number of slave automata invoked by $\nestedA$ before position $i$.
Observe that $\nestedA_1$ invokes a new slave automaton whenever $\nestedA$ does. Therefore, $\nestedA_1$ also invoked $n$ 
slave automata before position $i$. The partial average up to position $i$ in $\nestedA(\optWord)$ is the sum of values 
of all slave automata invoked before position $i$, while in  $\nestedA_1(\optWord)$ this is the sum of values
all slave automata invoked before position $i$ accumulate before position $j > i$, the fist position past $i$ when a new slave automaton is invoked.
Therefore, the difference between partial averages of $\nestedA(\optWord)$ and  $\nestedA_1(\optWord)$ up to position $i$, denoted by $\Delta$,
is the sum of weights slave automata invoked before the position $i$ accumulate past position $j$ multiplied by $\frac{1}{n}$.
Now, due to density of word $\optWord$, each slave automaton invoked before the position $i$ works for at most $\log(n)$ steps, therefore
the absolute accumulated value is at most $C \log(n)$, where $C$ is the maximal absolute weight in the slave automata of $\nestedA$.
Hence, $|\Delta| < \frac{1}{n} \cdot k \cdot C \cdot \log(n)$. Therefore, the partial averages of $\nestedA(\optWord)$ and $\nestedA_1(\optWord)$ converge. 
In consequence, $\nestedA(\optWord) = \nestedA_1(\optWord)$.

Due to Lemma~\ref{l:dense}, there exist a dense word $\optWord$ (w.r.t. $\nestedA$), which has the minimal value among all words. Then,
we have $\nestedA(\optWord) = \nestedA_1(\optWord)$. Therefore, $\inf_{\iterWord} \nestedA_1(\iterWord) \leq \inf_{\iterWord} \nestedA(\iterWord)$.
Conversely, we show that there exists a \emph{optimal} word $\optWord$ such that $\inf_{\iterWord} \nestedA_1(\iterWord) = \nestedA_1(\optWord)$ and 
$\optWord$ is a dense word (w.r.t. $\nestedA$). Notice that every time $\nestedA_1$ encodes in its generalized \buchi{} condition that all slave automata terminate infinitely often.
It follows that between two positions $a,b$ at which the generalized \buchi{} condition is satisfies, all slave automata invoked before $a$ terminate before $b$.
Thus, to show that $\optWord$ is a dense word (w.r.t. $\nestedA$), we show that the generalized \buchi{} condition of $\nestedA_1$ is satisfied at $\optWord$ sufficiently densely.
Observe that a deterministic $(\flimavg;\fsum)$-automaton of width bounded by $1$ is essentially
equivalent to a ${\flimavg}$-automaton with $\epsilon$-transitions (see Lemma~\ref{l:WidthOnePoly}), for which
one can construct words of optimal values which satisfies their (generalized) \buchi{} condition arbitrarily densely as long as the density converges to $0$.
Thus, such a word $\optWord$ exists, $\nestedA(\optWord) = \nestedA_1(\optWord)$. It follows that
$\inf_{\iterWord} \nestedA_1(\iterWord) = \inf_{\iterWord} \nestedA(\iterWord)$.
\end{proof}

\begin{claim}
The emptiness problem for deterministic $\flimavg$-automata with weights given in unary is $\NLOGSPACE$-complete.
\label{c:nlogspace}
\end{claim}
\begin{proof}
$\NLOGSPACE$-hardness readily follows from $\NLOGSPACE$-hardness of directed graph reachability.

For containment in $\NLOGSPACE$, recall that the infimum over all words of the values of a given $\flimavg$-automaton $\aut$
is less of equal to a given $\lambda$ if and only if there exists a cycle in the automaton such that
\begin{compactenum}[(1)]
\item the cycle is reachable firm the initial state, 
\item the sum of weights along the cycle is less of equal to $\lambda$, and
\item there exists an accepting state $s_a$ and a state $s_c$ such that $s_a$ is reachable from $s_c$ and vice versa.
\end{compactenum}
To check conditions (1),(2), (3) we non-deterministically pick states $s_a,s_c$ and verify conditions
(1) and (3) with reachability queries, which are in $\NLOGSPACE$. 
To check (2) we observe that if there exists a cycle satisfying conditions (1),(2), (3), then there also
exists a cycle that satisfy these conditions and of length bounded by $C \cdot |\aut|$, where $C$ is the maximal absolute value of the weights from $\aut$.
Now, if weights in $\aut$ are given in unary, $C < |\aut|$, the length of the cycle is at most $|\aut|^2$ and we need only logarithmic memory 
to non-deterministically pick the cycle state by state and keep track of the sum to verify that it is less or equal to $\lambda$.
\end{proof}

\begin{restatable}{lemmaStatement}{WidthOne}
The emptiness problem for deterministic $(\flimavg;\fsum)$-automata of width $1$ is 
in $\PTIME$ and if the weights are in unary, then it is in $\NLOGSPACE$.
\label{l:WidthOnePoly}
\end{restatable}
\begin{proof}
We observe that $\nestedA$ is essentially a deterministic $\flimavg$-automaton with silent moves. 
More precisely, let $Q_m$ be the set of states of the master automaton of $\nestedA$ and let $n$
be the number of slave automata. 
A run of the automaton of width bounded by $1$, can be partitioned into two types of fragments:
\begin{compactenum}[(1)]
\item fragments corresponding to a single run of slave automata, which are characterized by 
$q_1,q_2 \in Q_m$, the state of the master automaton at beginning and at the end of the fragment, and 
\item fragments where no slave automaton is running, i.e., only dummy slave automata are
invoked,  which are characterized by:

\begin{compactitem}
\item $q_1,q_2 \in Q_m$, the state of the master automaton at beginning and at the end of the fragment,
\item the first letter of the fragment $a$,  
\item the index $i \in \{1,\ldots, n\}$ of the slave automaton invoked, and
\item the value returned by the invoked slave automaton.
\end{compactitem}
\end{compactenum}

We consider a succinct representation of runs of $\nestedA$, where fragments of type (1) are replaced by a single letter $(q_1, q_2)$ and
fragments of type (2) are substituted by a single letter $(q_1,a, q_2,i)$. Moreover, we consider only maximal fragments, i.e., two fragments of type (1) can be merged into one fragment, therefore we forbid 
two successive occurrences of letters $(q_1,q_2) (q_2,q3)$.

Let $\Delta = \{ (q_1, q_2) : q_1, q_2 \in Q_m\} \cup \{ (q_1,a, q_2,i) : q_1, q_2 \in Q_m, a \in \Sigma, i \in \{1,\ldots, n\}\}$. 
We can define a deterministic $\flimavg$-automaton $\nonnestedA$ with silent moves over $\Delta$ which accepts only words that represent accepting runs of
$\nestedA$. The automaton $\nonnestedA$ checks that for two successive letters $a,b$ the second state in $a$ is the same as the first state in $b$ (e.g., $(q_1, q_2, i) (q_2, q_3) (q_3, q4, i')$), 
and it has a list of valid letters, i.e., a letter are valid if there exists a fragment corresponding to it. More precisely, $(q_1, q_2, i)$ is valid iff there exists a word $v$ such that
that (1)~the master automaton in the state $q_1$ upon reading letter $v[1]$ takes a transition at which it invokes $\slaveA_i$, 
(2)~the master automaton moves from $q_1$ to $q_2$ upon reading $v$, and (3)~slave automaton $\slaveA_i$ accepts the word $v$. Validity of letters $(q_1, q_2)$ is defined similarity.

Moreover, transitions over letters $(q_1, q_2)$ are silent (have no value) and transitions over letters
$(q_1, q_2,i)$ have the minimal value associated with such a fragment, i.e., it is the minimal value 
slave automaton $\slaveA_i$ can return on a word $v$ such that the master automaton moves on $v$ from state $q_1$ to $q_2$.
Thus, the value of $\nonnestedA$ on $\optWord' \in \Delta$ is (provided it is accepted) the minimal value $\nestedA$ can return
on the run with the sequence of fragments corresponding to $\optWord'$. Every run has the corresponding sequence of fragments 
In consequence,   $\inf_{\iterWord \in \Sigma^{\omega}}  \nestedA(\iterWord) = \inf_{\iterWord' \in \Delta^{\omega}}  \nestedA(\iterWord')$.

Finally, observe that $\nonnestedA$ can be constructed upon demand from $\nestedA$ in logarithmic space, i.e., we can answer query about parts of $\nonnestedA$ without outputting the whole automaton. 
The emptiness problem for $\flimavg$-automata is decidable in $\NLOGSPACE$ if weights are given in unary notation and in $\PTIME$ if weights are given in binary (Claim~\ref{c:nlogspace}). 
The silent moves are interleave with non-silent moves, therefore we can easily remove them and decide
the emptiness problem for $\nonnestedA$, and in turn the emptiness problem for $\nestedA$, (a)~in $\NLOGSPACE$ for unary weights, and
(b)~in $\PTIME$ for binary weights.
\end{proof}

\softParagraph{Key intuitions}. 
We show that every transition of $\masterA$, the master automaton of $\nestedA$, 
at which a slave automaton is invoked, can be substituted by a transition 
whose weight is the minimal value the invoked slave automaton can achieve. 
More precisely, while a slave automaton is running on the input word, 
the master automaton $\masterA$ is still active. 
Therefore, we substitute transitions $(q,a,q',i)$ of $\masterA$
by multiple transitions of the form $(q, (q,a,i,q''), q'')$, where 
$(q,a,i,q'')$ is a new letter, $q''$ is a state of $\masterA$
and the weight of this transition is the minimal value $\slaveA_i$ can achieve over words 
$au$ such that $\masterA$ moves from $q$ to $q''$ upon reading $au$.
Such a transformation preserves the infimum over all words and it transforms a deterministic 
$(\flimavg;\fsum)$-automaton of width $1$ to a deterministic $\flimavg$-automaton. 
The emptiness problem for $\flimavg$-automaton is decidable in $\PTIME$ and 
even in $\NLOGSPACE$ provided that weights are given in unary.
 
We now present the algorithm and lower bound for our main result.

\Paragraph{The algorithm}. We present an algorithm, which, given
a non-deterministic $(\flimavg,\fsum)$-automaton $\nestedA$ of width $k$ and $\lambda \in \Q$, decides whether $\inf_{\iterWord \in \Sigma^{\omega}} \nestedA(\iterWord) \leq \lambda$.
\begin{compactenum}
\item Transform $\nestedA$ into a deterministic $(\flimavg,\fsum)$-automaton $\nestedA_d$ of the same width such that
$\inf_{\iterWord \in \Sigma^{\omega}} \nestedA(\iterWord) =  \inf_{\iterWord \in (\Sigma \times \Gamma)^{\omega}} \nestedA_d(\iterWord)$ (Lemma~\ref{l:deterministic}).
\item Check condition (*) for  $\nestedA_d$. If it holds, then $\inf_{\iterWord \in \Sigma^{\omega}} \nestedA(\iterWord) = -\infty$ and return answer \textbf{YES}. 
Otherwise, continue the algorithm.
\item Transform $\nestedA_d$ into a deterministic $(\flimavg,\fsum)$-automaton $\nestedA_1$ of width $1$ such that
  $\inf_{\iterWord \in (\Sigma \times \Gamma)^{\omega}} \nestedA_d(\iterWord) = \inf_{\iterWord \in \Delta^{\omega}} \nestedA_1(\iterWord)$ (Lemma~\ref{l:reduction}).
\item Compute $\inf_{\iterWord \in \Delta^{\omega}} \nestedA_1(\iterWord)$ (Lemma~\ref{l:WidthOnePoly}), and return whether 
$\inf_{\iterWord \in \Delta^{\omega}} \nestedA_1(\iterWord) \leq \lambda$.
\end{compactenum}
Transformations in (1) and (3) are polynomial in the size of the automaton and exponential in $k$. Also, 
transformation from (1) does not increase $k$. Therefore, the size of $\nestedA_1$ is polynomial in the size $\nestedA$
and singly exponential in $k$. Moreover, these transformations can be done on-the-fly, i.e., 
there is not need to store the whole resulting automaton. 
Therefore, checks from (2) and (4), can be done in $\NLOGSPACE$ if $k$ is constant
and weights are in unary, $\PTIME$ if $k$ is constant,  
and $\PSPACE$ if $k$ is given in unary.

\Paragraph{Hardness results}.
If $k$ is constant, then the reachability problem on directed graphs, which is $\NLOGSPACE$-complete, can be reduced to language emptiness 
of a finite automaton, which is a special case the emptiness problem for non-deterministic 
$(\flimavg,\fsum)$-automata of width~1 with unary weights.
If $k$ is given in unary, consider the emptiness problem for the intersection of 
regular languages, which given $k$ and regular languages $\lang_1, \ldots, \lang_k$, asks
whether $\lang_1 \cap \ldots \cap \lang_k = \emptyset$. 
This problem is $\PSPACE$-complete~\cite{DBLP:conf/focs/Kozen77} and reduces to 
the emptiness problem for deterministic $(\flimavg,\fsum)$-automata of width 
given in unary: the $\PSPACE$-hardness result for emptiness of 
NWA given in~\cite{nested} uses NWA of width $|\nestedA|$.

\begin{theorem}
The emptiness problem for non-deterministic $(\flimavg,\fsum)$-automata is 
(a)~$\NLOGSPACE$-complete in the size of $\nestedA$ for constant width~$k$ with weights 
in unary; (b)~$\PTIME$ in the size of $\nestedA$ for constant width~$k$;
and 
(c)~$\PSPACE$-complete when the bounded width $k$ is given as input in unary.
\label{th:main}
\end{theorem}



%% file: main.bbl
\begin{thebibliography}{10}

\bibitem{AlmagorBK14}
Shaull Almagor, Udi Boker, and Orna Kupferman.
\newblock Discounting in {LTL}.
\newblock In {\em {TACAS}, 2014}, pages 424--439, 2014.

\bibitem{DBLP:conf/lics/AlurDDRY13}
Rajeev Alur, Loris D'Antoni, Jyotirmoy~V. Deshmukh, Mukund Raghothaman, and
  Yifei Yuan.
\newblock Regular functions and cost register automata.
\newblock In {\em LICS 2013}, pages 13--22, 2013.

\bibitem{Baier-CSL-LICS-1}
Christel Baier, Clemens Dubslaff, and Sascha Kl{\"{u}}ppelholz.
\newblock Trade-off analysis meets probabilistic model checking.
\newblock In {\em {CSL}-{LICS} 2014}, pages 1:1--1:10, 2014.

\bibitem{Baier-CSL-LICS-2}
Christel Baier, Joachim Klein, Sascha Kl{\"{u}}ppelholz, and Sascha Wunderlich.
\newblock Weight monitoring with linear temporal logic: complexity and
  decidability.
\newblock In {\em {CSL}-{LICS} 2014}, pages 11:1--11:10, 2014.

\bibitem{BokerCHK14}
Udi Boker, Krishnendu Chatterjee, Thomas~A. Henzinger, and Orna Kupferman.
\newblock Temporal specifications with accumulative values.
\newblock {\em {ACM} {TOCL}}, 15(4):27:1--27:25, 2014.

\bibitem{bollig2010pebble}
Benedikt Bollig, Paul Gastin, Benjamin Monmege, and Marc Zeitoun.
\newblock Pebble weighted automata and transitive closure logics.
\newblock In {\em {ICALP} 2010, Part {II}}, pages 587--598. Springer, 2010.

\bibitem{BouyerMM14}
Patricia Bouyer, Nicolas Markey, and Raj~Mohan Matteplackel.
\newblock Averaging in {LTL}.
\newblock In {\em {CONCUR} 2014}, pages 266--280, 2014.

\bibitem{BBCFK11}
Tom{\'{a}}s Br{\'{a}}zdil, V{\'{a}}clav Brozek, Krishnendu Chatterjee, Vojtech
  Forejt, and Anton{\'{\i}}n Kucera.
\newblock Two views on multiple mean-payoff objectives in {M}arkov decision
  processes.
\newblock In {\em {LICS} 2011}, pages 33--42, 2011.

\bibitem{Forejt}
Tom{\'{a}}s Br{\'{a}}zdil, Krishnendu Chatterjee, Vojtech Forejt, and
  Anton{\'{\i}}n Kucera.
\newblock Multigain: {A} controller synthesis tool for {MDPs} with multiple
  mean-payoff objectives.
\newblock In {\em {TACAS} 2015}, pages 181--187, 2015.

\bibitem{Cha07}
Krishnendu Chatterjee.
\newblock Markov decision processes with multiple long-run average objectives.
\newblock In {\em FSTTCS}, pages 473--484, 2007.

\bibitem{CD11}
Krishnendu Chatterjee and Laurent Doyen.
\newblock Energy and mean-payoff parity {M}arkov {D}ecision {P}rocesses.
\newblock In {\em {MFCS} 2011}, pages 206--218, 2011.

\bibitem{Chatterjee:2009:AWA:1789494.1789497}
Krishnendu Chatterjee, Laurent Doyen, and Thomas~A. Henzinger.
\newblock Alternating weighted automata.
\newblock In {\em FCT'09}, pages 3--13. Springer, 2009.

\bibitem{DBLP:journals/corr/abs-1007-4018}
Krishnendu Chatterjee, Laurent Doyen, and Thomas~A. Henzinger.
\newblock Expressiveness and closure properties for quantitative languages.
\newblock {\em LMCS}, 6(3), 2010.

\bibitem{Chatterjee08quantitativelanguages}
Krishnendu Chatterjee, Laurent Doyen, and Thomas~A. Henzinger.
\newblock Quantitative languages.
\newblock {\em ACM TOCL}, 11(4):23, 2010.

\bibitem{CFW13}
Krishnendu Chatterjee, Vojtech Forejt, and Dominik Wojtczak.
\newblock Multi-objective discounted reward verification in graphs and {MDPs}.
\newblock In {\em {LPAR}}, pages 228--242, 2013.

\bibitem{nested}
Krishnendu Chatterjee, Thomas~A. Henzinger, and Jan Otop.
\newblock Nested weighted automata.
\newblock In {\em {LICS} 2015}, pages 725--737, 2015.

\bibitem{ChatterjeeHO15}
Krishnendu Chatterjee, Thomas~A. Henzinger, and Jan Otop.
\newblock Quantitative automata under probabilistic semantics.
\newblock {\em CoRR}, abs/1604.06764, 2016.
\newblock A conference version accepted to LICS 2016.

\bibitem{CKK15}
Krishnendu Chatterjee, Zuzana Kom{\'{a}}rkov{\'{a}}, and Jan
  Kret{\'{\i}}nsk{\'{y}}.
\newblock Unifying two views on multiple mean-payoff objectives in {M}arkov
  {D}ecision {P}rocesses.
\newblock In {\em {LICS} 2015}, pages 244--256, 2015.

\bibitem{CMH06}
Krishnendu Chatterjee, Rupak Majumdar, and Thomas~A. Henzinger.
\newblock {M}arkov {D}ecision {P}rocesses with multiple objectives.
\newblock In {\em {STACS} 2006}, pages 325--336, 2006.

\bibitem{Droste:2009:HWA:1667106}
Manfred Droste, Werner Kuich, and Heiko Vogler.
\newblock {\em Handbook of Weighted Automata}.
\newblock Springer, 1st edition, 2009.

\bibitem{DrosteR06}
Manfred Droste and George Rahonis.
\newblock Weighted automata and weighted logics on infinite words.
\newblock In {\em {DLT} 2006}, pages 49--58, 2006.

\bibitem{filar}
Jerzy Filar and Koos Vrieze.
\newblock {\em Competitive Markov decision processes}.
\newblock Springer, 1996.

\bibitem{FKN11}
Vojtech Forejt, Marta~Z. Kwiatkowska, Gethin Norman, David Parker, and Hongyang
  Qu.
\newblock Quantitative multi-objective verification for probabilistic systems.
\newblock In {\em TACAS}, pages 112--127, 2011.

\bibitem{modelmeasuring}
Thomas~A. Henzinger and Jan Otop.
\newblock From model checking to model measuring.
\newblock In {\em {CONCUR} 2013}, pages 273--287, 2013.

\bibitem{DBLP:conf/focs/Kozen77}
Dexter Kozen.
\newblock Lower bounds for natural proof systems.
\newblock In {\em FOCS}, pages 254--266. IEEE Computer Society, 1977.

\bibitem{DBLP:journals/jalc/Mohri02}
Mehryar Mohri.
\newblock Semiring frameworks and algorithms for shortest-distance problems.
\newblock {\em J. Aut. Lang. \& Comb.}, 7(3):321--350, 2002.

\bibitem{Puterman}
Martin~L. Puterman.
\newblock {\em {M}arkov {D}ecision {P}rocesses: Discrete Stochastic Dynamic
  Programming}.
\newblock Wiley, 1st edition, 1994.

\end{thebibliography}
